\documentclass[a4paper,11pt]{article}
\usepackage{hyperref}
\usepackage[top=1 in,bottom=1 in,left=1.0 in,right=1.0 in]{geometry}
\hypersetup{hidelinks}
\usepackage{amsthm}
\usepackage{rotating}
\usepackage{amssymb,amsmath}

\numberwithin{equation}{section}

\newtheorem{lemma}{Lemma}
\newtheorem{proposition}{Proposition}

\newtheorem{theorem}{Theorem}
\newtheorem{remark}{Remark}

\newcommand{\bs}{\boldsymbol}
\newcommand{\mb}{\mathbb}
\newcommand{\mf}{\mathbf}
\newcommand{\mr}{\mathrm}
\newcommand{\e}{\mathbf{e}}
\newcommand{\I}{\mathrm{i}}

\newcommand{\bsb}{\begin{subequations}}
\newcommand{\esb}{\end{subequations}}

\newcommand{\ba}[1]{\bar{#1}}

\usepackage{amsmath}
\usepackage{amsfonts}
\usepackage{latexsym}
\usepackage{amssymb}
\usepackage{appendix}
\usepackage{mathrsfs}
\usepackage{graphicx}
\usepackage{bm}
\usepackage{arydshln}
\usepackage{tikz}

\usepackage{color}
\newcommand{\bred}{\begin{color}{red}}
\newcommand{\ecl}{\end{color}}
\newcommand{\bblue}{\begin{color}{blue}}
\newcommand{\bgre}{\begin{color}{green}}
\newcommand{\bora}{\begin{color}{orange}}

\begin{document}

\title{\textbf{Solutions to the SU($\mathcal{N}$) self-dual Yang--Mills equation}}

\author{Shangshuai Li$^1$,~~ Changzheng Qu$^2$,~~
Da-jun Zhang$^1$\footnote{Corresponding author. Email: djzhang@staff.shu.edu.cn} \\
{\small $~^1$Department of Mathematics, Shanghai University, Shanghai 200444,  China}\\
{\small $~^2$School of Mathematics and Statistics, Ningbo University, Ningbo 315211, China}}
\maketitle

\begin{abstract}
In this paper we aim to derive solutions for the SU($\mathcal{N}$) self-dual Yang--Mills
(SDYM) equation with arbitrary $\mathcal{N}$.
A set of noncommutative relations are introduced to construct a matrix equation
that can be reduced to the  SDYM equation.
It is shown that these relations can be generated from two different Sylvester equations,
which correspond to the two Cauchy matrix schemes for the (matrix) Kadomtsev--Petviashvili hierarchy
and the (matrix) Ablowitz--Kaup--Newell--Segur hierarchy, respectively.
In each Cauchy matrix scheme we investigate the possible reductions
that can lead to the SU$(\mathcal{N})$  SDYM equation
and also analyze the physical significance of some solutions,
i.e. being Hermitian, positive-definite and of determinant being one.

\begin{description}
\item[Keywords:] self-dual Yang--Mills equation, Cauchy matrix approach, Sylvester equation,
exact solution, integrable system
\item[PACS numbers:] 02.30.Ik, 02.30.Ks, 05.45.Yv
\end{description}
\end{abstract}

\section{Introduction}\label{sec-1}

The SU$(\mathcal{N})$ self-dual Yang--Mills (SDYM) equation takes a form \cite{BFNY-1978}
\begin{equation}\label{SDYM}
     (J_{\bar y}J^{-1})_y+(J_{\bar z}J^{-1})_z=0,
\end{equation}
where $J$ is a $\mathcal{N} \times \mathcal{N}$ matrix function of
$(y, \ba y, z, \ba z)\in \mathbb{C}^4$.
To agree with the metric in the Euclidean space $\mathbb{R}^4$, it requires that $\ba y$ and $\ba z$
are complex conjugates of $y$ and $z$, and $J$ is a
positive-definite and Hermitian matrix with $|J|=1$ (see Yang's formulation \cite{Yang-1977}).
Note that the original Yang--Mills equation \cite{YM-1954} is not integrable in general
but the self-dual gauge field equations are integrable, e.g., in the sense of Painlev\'e property \cite{JKM-1982,Ward-1984}.
One can also refer to an early review \cite{P-1980} or a recent paper \cite{LQYZ-SAPM-2022}
and the references therein, or a recent theses \cite{Huang-PhD}, for more details.

In the recent paper \cite{LQYZ-SAPM-2022} we developed an approach to construct solutions to the
SU$(2)$ SDYM equation by means of the Cauchy matrix method.
It is based on a framework of the  Cauchy matrix approach
to the Ablowitz--Kaup--Newell--Segur (AKNS) equations.
Through introducing an independent  variable $x_0$,
we were able to establish a set of noncommutative recursive and evolution relations
for  master functions $\{\bs S^{(i,j)}\}$
(see equations (22) and (23) in \cite{LQYZ-SAPM-2022}),
which enable us to construct solutions to the SU$(2)$ SDYM equation.
It is possible to extend the approach developed in \cite{LQYZ-SAPM-2022} to
the SU$(\mathcal{N})$ SDYM equation for $\mathcal{N}$ being even, however, not for arbitrary $\mathcal{N}$.

In the present paper
we  aim to construct solutions for the SU$(\mathcal{N})$ SDYM equation with arbitrary $\mathcal{N}$.
Our plan is the following.
We will begin by
recalling the Cauchy matrix approach used in  \cite{LQYZ-SAPM-2022} for the SU$(2)$ SDYM equation and
presenting a set of  derivative and difference relations
for the master functions $\bs S^{(i,j)}$ (see  \eqref{Sij_reqq} and \eqref{Sij_movesqq}).
We naively assume these relations can be extended from $2\times 2$ to
$\mathcal{N}\times \mathcal{N}$ matrices with arbitrary $\mathcal{N}$.
Then we will look for $\bs S^{(i,j)}$ that can match these relations.
It turns out that we can have two Cauchy matrix schemes to define such $\bs S^{(i,j)}$,
which are associated with two different Sylvester equations,
corresponding to the (matrix) Kadomtsev--Petviashvili (KP) hierarchy
and the (matrix) AKNS hierarchy, respectively.
We will investigate how solutions of the SU$(\mathcal{N})$  SDYM equations
are constructed from these two Cauchy matrix schemes
by imposing constraints.
As examples, we will discuss  some explicit solutions of the SU$(2)$ SDYM equation
with respect to their physical significance.

The paper is organized as follows. In Section \ref{sec-2}
we  recall the Cauchy matrix approach for the SU$(2)$ SDYM equation
and introduce a set of noncommutative evolution relations and a recursive relation
for the master matrix functions $\{\bs S^{(i,j)}\}$,
and then we give two different Sylvester equations that can be used to define the master functions
subject to these relations. In Section \ref{sec-3} and \ref{sec-4}
we investigate the two Sylvester equations and their related
SU$(\mathcal{N})$ SDYM equations. At this stage the investigation of solutions will be
elaborated about their determinants and the Hermitian property.
We also analyze some solutions for their positive-definiteness.
Finally in Section \ref{sec-5} concluding remarks are given.
There are two appendices  where in Appendix \ref{APP-A} we present
an alternative proof to obtain equation \eqref{diff_recur}, 
and in Appendix \ref{APP-B} we provide an example of dimension reduction
to a 3-dimensional system.

\section{Cauchy matrix schemes}\label{sec-2}

\subsection{Cauchy matrix approach to the SU$(2)$ SDYM equation}\label{sec-2-1}

Let us briefly recall the Cauchy matrix approach used in \cite{LQYZ-SAPM-2022}
for deriving solutions to the SU$(2)$ SDYM equation.
The Cauchy matrix approach is a method to construct and study integrable equations by
means of the Sylvester-type equations. In this approach integrable equations are presented as
closed forms of some functions (i.e. $\{\bs S^{(i,j)}\}$) and their shifts (or derivatives).
It is first systematically used in \cite{NAH-2009} to
investigate ABS equations and later developed in \cite{XZZ-2014,ZZ-2013} to more general cases.

Consider the Sylvester equation
\begin{align}\label{Syl_eqqq}
    \bs K\bs M-\bs M\bs K=\bs r\bs s^T,
\end{align}
assigned with general dispersion relations
\begin{align}\label{rs_moveqq}
    \bs r_{x_n}=\bs{AK}^n\bs r,~~ \bs s_{x_n}=\bs A(\bs K^T)^n\bs s, ~~(n\in\mb Z),
\end{align}
where $\{x_n\}$ are infinitely many complex independent variables,
$\bs K,\bs M,\bs A,\bs r,\bs s$ are block matrices
\begin{align*}
    \bs K=
    \begin{pmatrix}
        \bs K_1 & \bs 0 \\
        \bs 0 & \bs K_2 \\
    \end{pmatrix},
    ~
    \bs M=
    \begin{pmatrix}
        \bs 0 & \bs M_1 \\
        \bs M_2 & \bs 0 \\
    \end{pmatrix},
~
    \bs A=
    \begin{pmatrix}
        \bs I_{N_1} & \bs 0 \\
        \bs 0 & -\bs I_{N_2}
    \end{pmatrix},
~
    \bs r=
    \begin{pmatrix}
        \bs r_1 & \bs 0 \\
        \bs 0 & \bs r_2 \\
    \end{pmatrix},
 ~
    \bs s=
    \begin{pmatrix}
        \bs 0 & \bs s_1 \\
        \bs s_2 & \bs 0 \\
    \end{pmatrix},
\end{align*}
with $\bs K_i\in \mb C_{N_i\times N_i},\bs M_1\in \mb C_{N_1\times N_2}[{\mathbf{x}}],
\bs M_2\in\mb C_{N_2\times N_1}[\mathbf{x}], \bs r_i,\bs s_i\in \mb C_{N_i\times 1}[\mathbf{x}]$,
$\bs I_{N_i}$ being the $N_i$-th order identity matrix,
$\mathbf{x}=(\cdots, x_{-1}, x_0, x_1,\cdots)$, and $N_1+N_2=2N$.
It is also assumed that $\bs K_1$ and $\bs K_2$ are invertible and do not share any eigenvalues
so that  for given $\bs K,  \bs r,\bs s$
the Sylvester equation \eqref{Syl_eqqq} has a unique solution $\bs M$ \cite{Syl}.
Define  $2 \times 2$ matrix functions
\begin{align}\label{Sijqq}
    \bs S^{(i,j)}= \bs s^T\bs K^j(\bs I_{2N}+\bs M)^{-1}\bs K^i\bs r, ~~ (i,j\in \mathbb{Z}).
\end{align}
Using the Sylvester equation \eqref{Syl_eqqq} one can derive a recursive relation for $\{\bs S^{(i,j)}\}$
(cf.\cite{XZZ-2014,Zhao-2018}):
\begin{align}\label{Sij_reqq}
    \bs S^{(0,j)}\bs S^{(i,0)}=\bs S^{(i+1,j)}-\bs S^{(i,j+1)},
\end{align}
which is independent of the dispersion relation \eqref{rs_moveqq}.
By differentiate Sylvester equation \eqref{Syl_eqqq} with respect to $x_n$ and making use of \eqref{rs_moveqq},
one can derive evolution relations of $\{\bs S^{(i,j)}\}$
(see \cite{LQYZ-SAPM-2022}):
\begin{subequations}\label{Sij_movesqq}
\begin{align}
\bs S^{(i,j)}_{x_n}&=\bs S^{(i+n,j)}\bs a-\bs a\bs S^{(i,j+n)}
    -\sum_{l=0}^{n-1}\bs S^{(n-1-l,j)}\bs a\bs S^{(i,l)}, &(n\in\mb Z^+),  \\
\bs S^{(i,j)}_{x_0}&=\bs S^{(i,j)}\bs a-\bs a\bs S^{(i,j)}=[\bs S^{(i,j)},\bs a], &~ \\
\bs S^{(i,j)}_{x_n}&=\bs S^{(i+n,j)}\bs a-\bs a\bs S^{(i,j+n)}
    +\sum_{l=-1}^{n}\bs S^{(n-1-l,j)}\bs a\bs S^{(i,l)}, &(n\in\mb Z^-),
\end{align}
\end{subequations}
where $\bs a=\mf{diag}(1,-1)$.
We call $\{\bs S^{(i,j)}\}$ master functions in the Cauchy matrix approach
since nonlinear equations arise as closed forms of them (see \cite{NAH-2009,XZZ-2014,ZZ-2013}).
These relations enabled us to obtain the equations \cite{LQYZ-SAPM-2022}
\begin{align}\label{recur_uvqq}
        \bs v_{x_{n+1}}\bs v^{-1}=-\bs u_{x_n}, ~~ (n\in\mb Z),
\end{align}
where  $\bs u=\bs S^{(0,0)}$ and $\bs v= \bs I_2-\bs S^{(-1,0)}$,
and consequently,
\begin{align}\label{SDYM-30qq}
    (\bs v_{x_{n+1}}\bs v^{-1})_{x_m}-(\bs v_{x_{m+1}}\bs v^{-1})_{x_n}=0,
    ~~(n,m\in \mathbb{Z}).
\end{align}
By properly imposing constraints on the coordinates $\{x_n\}$ we were able to recover the SU$(2)$ SDYM equation
\eqref{SDYM} from  \eqref{SDYM-30qq}.
More details can be referred from \cite{LQYZ-SAPM-2022}.

\subsection{General assumptions of derivative and difference relation}\label{sec-2-2}

The relations \eqref{Sij_reqq} and \eqref{Sij_movesqq},
which hold for  the $2\times 2$ matrix case,
are crucial and important to obtaining equation \eqref{recur_uvqq}
that finally leads us to the  SU$(2)$ SDYM equation.
However, the scheme used in \cite{LQYZ-SAPM-2022} can only be
extended to the case of $\mathcal{N}$ being even.
In this paper, we naively assume that such a set of relations hold as well for
$\mathcal{N}\times \mathcal{N}$ matrix functions $\{\bs S^{(i,j)}\}$
with arbitrary $\mathcal{N}$.
This will immediately give rise to the equation \eqref{recur_uvqq} and so \eqref{SDYM-30qq}
where both $\bs u$ and $\bs v$ are $\mathcal{N}\times \mathcal{N}$ matrix functions,
from which we might find solutions to the SU$(\mathcal{N})$ SDYM equation \eqref{SDYM}.
Let us begin by stating our assumption.

\begin{lemma}\label{L-2-1}
    Suppose there exist  matrix functions $\bs S^{(i,j)}\in\mb C_{\mathcal{N}\times\mathcal{N}}[\mf x]$
    that satisfy the evolution relations
    \begin{subequations}\label{Sij_moves}
    \begin{align}
    \bs S^{(i,j)}_{x_n}&=\bs S^{(i+n,j)}\bs a-\bs a\bs S^{(i,j+n)}
    -\sum_{l=0}^{n-1}\bs S^{(n-1-l,j)}\bs a\bs S^{(i,l)}, &(n\in\mb Z^+),\label{Sij_moves1} \\
    \bs S^{(i,j)}_{x_0}&=\bs S^{(i,j)}\bs a-\bs a\bs S^{(i,j)}=[\bs S^{(i,j)},\bs a], &~\label{Sij_moves2}\\
    \bs S^{(i,j)}_{x_n}&=\bs S^{(i+n,j)}\bs a-\bs a\bs S^{(i,j+n)}
    +\sum_{l=-1}^{n}\bs S^{(n-1-l,j)}\bs a\bs S^{(i,l)}, &(n\in\mb Z^-) \label{Sij_moves3}
    \end{align}
    \end{subequations}
    and the difference relation
    \begin{align}\label{diff_sum_pro}
        \bs S^{(i+1,j)}-\bs S^{(i,j+1)}=\bs S^{(0,j)}\bs S^{(i,0)},
    \end{align}
    where $i,j\in \mathbb{Z}$,
    $\mf x\doteq(\cdots, x_{-1}, x_0, x_{1},\cdots)$,
    $\bs a=\mr{diag}(a^{(1)},\cdots,a^{(\mathcal{N})})$ with $a^{(i)}\in \mathbb{C}$,
    and $[A,B]=AB-BA$.
    Define
    \begin{equation}\label{UV}
        \bs U\doteq \bs S^{(0,0)},  ~~\bs V \doteq \bs I_{\mathcal{N}}-\bs S^{(-1,0)}.
    \end{equation}
    Then we have
    \begin{equation}\label{diff_recur}
        \bs V_{x_{n+1}}\bs V^{-1}=-\bs U_{x_n}, ~~  (n\in \mb Z),
    \end{equation}
    and consequently,
    \begin{align}\label{SDYM-3}
    (\bs V_{x_{n+1}}\bs V^{-1})_{x_m}-(\bs V_{x_{m+1}}\bs V^{-1})_{x_n}=0, ~~ (n,m\in\mb Z).
    \end{align}
\end{lemma}

The proof is the same as the one for Theorem 2 in \cite{LQYZ-SAPM-2022} for the $2\times 2$ case. We skip it.

\subsection{Sylvester equations for equation \eqref{SDYM-3}}\label{sec-2-3}

Next, we need to look for well-defined $\mathcal{N}\times \mathcal{N}$ matrix functions $\{\bs S^{(i,j)}\}$
that can meet the assumed relations \eqref{Sij_moves} and \eqref{diff_sum_pro}.
To achieve that, we introduce a general Sylvester equation
\begin{align}\label{Syl_eq}
        \bs K\bs M-\bs M\bs L=\bs r\bs s^T,
\end{align}
where $\bs M\in\mb C_{N\times N}[\mf x]$, $\bs K,\bs L\in\mb C_{N\times N}$,
$\bs r,\bs s\in\mb C_{N\times\mathcal{N}}[\mf x]$
(one should be aware of the difference between $N$ and $\mathcal{N}$; here $N$ indicates the number of solitons).
Besides, we also require $\bs K,\bs L$ satisfy certain conditions such that
equation \eqref{Syl_eq} has a unique solution $\bs M$ for given $\bs K, \bs L, \bs r$ and $\bs s$.
Then we assign evolutions for  $\bs r$ and $\bs s$ as
\begin{align}\label{rs_move}
        \bs r_{x_n}=\bs K^n\bs r\bs a, ~~\bs s_{x_n}=-(\bs L^T)^n\bs s\bs a, ~~(n\in\mb Z),
\end{align}
where $\bs a$ is the diagonal matrix defined as in Lemma \ref{L-2-1}.
The above equations allow solutions
\begin{align}\label{rs_exp}
    \bs r=(\bs r^{(1)},\bs r^{(2)},\cdots,\bs r^{(\mathcal{N})}),
    ~~\bs s=(\bs s^{(1)},\bs s^{(2)},\cdots,\bs s^{(\mathcal{N})}),
    \end{align}
where $\bs r^{(i)}$ and $\bs s^{(i)}$ are $N$-th order column vectors expressed by
\begin{align*}
        \bs r^{(i)}=\exp\Big(a^{(i)}\sum_{n\in\mb Z}\bs K^nx_n\Big)\overset{\circ}{\bs r}{}^{(i)},~~~~
        \bs s^{(i)}=\exp\Big(-a^{(i)}\sum_{n\in\mb Z}(\bs L^T)^nx_n\Big)\overset{\circ}{\bs s}{}^{(i)},
\end{align*}
with $\overset{\circ}{\bs r}{}^{(i)}$ and $\overset{\circ}{\bs s}{}^{(i)}$ being
$N$-th order  constant column vectors.

\begin{remark}\label{R-2-0}
It is important that $\bs a$ is not an identity matrix.
Recalling relation (21) and (23) in \cite{LQYZ-SAPM-2022},
$\bs S^{(i,j)}$ will be independent of  $x_n$ if $\bs a$ is an identity matrix.
\end{remark}

Using the elements in the Sylvester equation \eqref{Syl_eq}, together with
$\bs r$ and $\bs s$ defined by \eqref{rs_move},
we introduce $\mathcal{N}\times\mathcal{N}$ matrix functions $\{\bs S^{(i,j)}\}$:
\begin{align}\label{Sij}
        \bs S^{(i,j)} \doteq \bs s^T\bs L^j(\bs C +\bs M)^{-1}\bs K^i\bs r, ~~ i,j\in \mathbb{Z},
\end{align}
where $\bs C$ is an arbitrary $N\times N$ complex matrix independent of $\mf x$.
%We call $\bs S^{(i,j)}$ the master functions in the Cauchy matrix approach,
%since nonlinear equations arise as closed forms of $\bs S^{(i,j)}$ (see \cite{NAH-2009,XZZ-2014,ZZ-2013}).
For their evolutions with respect to $x_n$, we have the following result.

\begin{proposition}\label{P-2-1}
Suppose that $\bs r$ and $\bs s$ are defined by \eqref{rs_move},
and the Sylvester equation \eqref{Syl_eq} has a unique solution $\bs M$ for given $\bs K$ and $\bs L$.
Then, the $\mathcal{N}\times\mathcal{N}$ matrix functions $\bs S^{(i,j)}$ defined in \eqref{Sij}
evolve as \eqref{Sij_moves}.
\end{proposition}

\begin{proof}
Differentiating the Sylvester equation \eqref{Syl_eq} with respect to $x_n$ and making use of
\eqref{rs_move}, for the case $n>0$ we have
\begin{align*}
\bs K \bs M_{x_n}- \bs M_{x_n}\bs L
& = \bs K^{n}\bs r\bs a\bs s^T- \bs r\bs a\bs s^T\bs L^{n}\\
& = \bs K \left(\sum^{n-1}_{l=0}\bs K^{n-1-l}\bs r\bs a\bs s^T\bs L^{l}\right)
       -\left(\sum^{n-1}_{l=0}\bs K^{n-1-l}\bs r\bs a\bs s^T\bs L^{l}\right)\bs L.
\end{align*}
This yields
\[\bs K \left(\bs M_{x_n}-\sum^{n-1}_{l=0}\bs K^{n-1-l}\bs r\bs a\bs s^T\bs L^{l}\right)
-\left( \bs M_{x_n}- \sum^{n-1}_{l=0}\bs K^{n-1-l}\bs r\bs a\bs s^T\bs L^{l}\right) \bs L=0. \]
With the assumption that equation \eqref{Syl_eq} has a unique solution, we have
\bsb\label{M_evo}
\begin{align}
&\bs M_{x_n}=\sum^{n-1}_{l=0}\bs K^{n-1-l}\bs r\bs a\bs s^T\bs L^{l},
~~~~  (n\in \mb Z^+). \label{M_evo1}
\end{align}
Similarly, we also get
\begin{align}
&\bs M_{x_0}=\bs 0,    \label{M_evo2}\\
&\bs M_{x_n}=-\sum^{n}_{l=-1}\bs K^{n-1-l}\bs r\bs a\bs s^T\bs L^{l},
~~~~ (n\in \mb Z^-). \label{M_evo3} \end{align}
\esb
Meanwhile, for the $\bs S^{(i,j)}$  defined in \eqref{Sij}, we find that
\[
\bs S^{(i,j)}_{x_n}=\bs s^T_{x_n}\bs L^j(\bs C+\bs M)^{-1}\bs K^i\bs r
+\bs s^T\bs L^j(\bs C+\bs M)^{-1}\bs K^i\bs r_{x_n}
 +\bs s^T\bs L^j((\bs C+\bs M)^{-1})_{x_n}\bs K^i\bs r .
 \]
Inserting  \eqref{rs_move} into the above equation yields
\[ \bs S^{(i,j)}_{x_n}
=-\bs a\bs S^{(i,j+n)}+\bs S^{(i+n,j)}\bs a
  -\bs s^T\bs L^j(\bs C+\bs M)^{-1}\bs M_{x_n}(\bs C+\bs M)^{-1}\bs K^i\bs r.\]
Then,
substituting \eqref{M_evo}
and expressing it in terms of $\bs S^{(i,j)}$,
we arrive at \eqref{Sij_moves}.

\end{proof}

For the difference relation \eqref{diff_sum_pro}, we find the following.

\begin{proposition}\label{P-2-2}
Suppose that  the Sylvester equation \eqref{Syl_eq} has a unique solution $\bs M$ for given $\bs K$ and $\bs L$.
For the matrix functions $\bs S^{(i,j)}$ defined in \eqref{Sij}
the difference relation \eqref{diff_sum_pro} holds provided
\begin{equation}\label{KC=CL}
\bs K\bs C-\bs C\bs L=\bs0.
\end{equation}
\end{proposition}

\begin{proof}
For $\bs S^{(i,j)}$ defined in \eqref{Sij}, it follows that
\begin{align*}
        \bs S^{(i+1,j)}-\bs S^{(i,j+1)}
 =~&\bs s^T\bs L^j(\bs C+\bs M)^{-1}\bs K^{i+1}\bs r
        -\bs s^T\bs L^{j+1}(\bs C+\bs M)^{-1}\bs K^i\bs r \\
 =~&\bs s^T\bs L^j(\bs C+\bs M)^{-1}\bs K(\bs C+\bs M)
 (\bs C+\bs M)^{-1}\bs K^i\bs r\\
   & ~~   -\bs s^T\bs L^{j}(\bs C+\bs M)^{-1}(\bs C+\bs M)
   \bs L(\bs C+\bs M)^{-1}\bs K^i\bs r \\
 =~&\bs s^T\bs L^{j}(\bs C+\bs M)^{-1}(\bs K\bs C-\bs C\bs L
 +\bs r\bs s^T)(\bs C+\bs M)^{-1}\bs K^i\bs r \\
 =~&\bs s^T\bs L^{j}(\bs C+\bs M)^{-1}(\bs K\bs C-\bs C\bs L)(\bs C+\bs M)^{-1}\bs K^i\bs r
 +\bs S^{(0,j)}\bs S^{(i,0)}.
\end{align*}
This indicates that, only when \eqref{KC=CL} holds,
we have the difference relation \eqref{diff_sum_pro}.

\end{proof}

Based on the above two propositions,  in what follows we identify $\bs K$ and $\bs L$ such that equation \eqref{KC=CL}
holds and meanwhile the Sylvester equation \eqref{Syl_eq} has a unique solution $\bs M$.
We are then led to the following two cases.
\begin{itemize}
  \item [1)]
  When $\bs K$ and $\bs L$ do not share eigenvalues,
  both equations \eqref{KC=CL} and \eqref{Syl_eq} have a unique solution (see \cite{Syl}),
  and in particular, $\bs C=\bs0$.
  \item [2)]
  When $\bs K=\bs L$, first,  \eqref{KC=CL} holds if $\bs K$ and $\bs C$ commute.
  For the Sylvester equation, we assume
 $\mathcal{N}=2\mathcal{M}, N=2M$,  which  are even numbers,
  and assume $\bs K$, $\bs M$, $\bs C$ and so forth  are block matrices with the following forms\footnote{
  Note that in this case, $M$ rather than $N$ represents the number of solitons.}
  \bsb\label{KAM}
  \begin{align}
    \bs K&=
    \begin{pmatrix}
        \bs K_1 & \bs 0 \\
        \bs 0 & \bs K_2 \\
    \end{pmatrix},
    &
    \bs M&=
    \begin{pmatrix}
        \bs 0 & \bs M_1 \\
        \bs M_2 & \bs 0 \\
    \end{pmatrix},
    &
    \bs C&=
    \begin{pmatrix}
        \bs C_1 & \bs 0 \\
        \bs 0 & \bs C_2 \\
    \end{pmatrix},
\\
    \bs a&=
    \begin{pmatrix}
        \bs a_1 & \bs 0 \\
        \bs 0 & \bs a_2
    \end{pmatrix},
    &
    \bs r&=
    \begin{pmatrix}
        \bs r_1 & \bs 0 \\
        \bs 0 & \bs r_2 \\
    \end{pmatrix},
    &
    \bs s&=
    \begin{pmatrix}
        \bs 0 & \bs s_1 \\
        \bs s_2 & \bs 0 \\
    \end{pmatrix},
\end{align}
\esb
where $\bs K_i,\bs C_i\in\mb C_{M\times M},\bs r_i,\bs s_i\in\mb C_{M\times\mathcal{M}}[\mf x],
\bs a_i\in\mb C_{\mathcal{M}\times\mathcal{M}}$,
$\bs M_i \in\mb C_{M\times M}[\mf x]$.
In this case, the Sylvester equation \eqref{Syl_eq} is decoupled into
\begin{subequations}\label{symm-case}
\begin{align}
&\bs K_1 \bs M_1-\bs M_1 \bs K_2=\bs r_1 \bs s^T_2,\\
&\bs K_2 \bs M_2-\bs M_2 \bs K_1=\bs r_2 \bs s^T_1,
\end{align}
\end{subequations}
and the dispersion relation \eqref{rs_move} is decoupled into
\begin{subequations}\label{DR-symm}
\begin{align}
& \bs r_{1,x_n}=\bs K_1^n \bs r_1 \bs a_1,  ~~
\bs s_{1,x_n}=-(\bs K_1^T)^n \bs s_1 \bs a_2, \label{DR-symm-1}\\
& \bs r_{2,x_n}=\bs K_2^n \bs r_2 \bs a_2,  ~~
\bs s_{2,x_n}=-(\bs K_2^T)^n \bs s_2 \bs a_1. \label{DR-symm-2}
\end{align}
\end{subequations}
We further assume  that $\bs K_1$ and $\bs K_2$ do not share eigenvalues
so that the coupled system \eqref{symm-case} determine unique solutions
$\bs M_1$ and $\bs M_2$.

\end{itemize}

Let us denote the above two cases as the case of  asymmetric Sylvester equation (with $\bs K\neq \bs L$)
and the case of  symmetric Sylvester equation (with $\bs K= \bs L$),
and summarize them in the following theorem.

\begin{theorem}\label{T-2-1}
The master functions $\{\bs S^{(i,j)}\}$ in the two cases agree with the derivative relations in
\eqref{Sij_moves} and difference relation \eqref{diff_sum_pro}.
\begin{itemize}
  \item [(1)] Asymmetric Sylvester equation case:
  The master functions $\{\bs S^{(i,j)}\}$ defined by \eqref{Sij} with $\bs C =\bs0$, i.e.
  \begin{align*}%\label{Sij-C0}
        \bs S^{(i,j)} = \bs s^T\bs L^j \bs M^{-1}\bs K^i\bs r, ~~ i,j\in \mathbb{Z},
  \end{align*}
  satisfy the relations \eqref{Sij_moves} and   \eqref{diff_sum_pro},
  when $\bs K, \bs L, \bs M, \bs r, \bs s$ satisfy  the (asymmetric) Sylvester equation \eqref{Syl_eq}
  and dispersion relation \eqref{rs_move}
  where $\bs K$ and $\bs L$ do not share eigenvalues.
  \item [(2)]  Symmetric Sylvester equation case:
  The master functions $\{\bs S^{(i,j)}\}$ defined by \eqref{Sij}
  meet the requirements  \eqref{Sij_moves} and  \eqref{diff_sum_pro},
  when the following conditions are satisfied:
  $\bs K=\bs L$, the elements in the Sylvester equation \eqref{Syl_eq}
  and dispersion relation \eqref{rs_move} take forms as in \eqref{KAM}
  and obey the coupled Sylvester equations \eqref{symm-case} and dispersion relation \eqref{DR-symm},
  and $\bs K_i \bs C_i=\bs C_i \bs K_i$ for $i=1,2$.
\end{itemize}
In each case, $\bs U$ and $\bs V$ defined by \eqref{UV} satisfy equation \eqref{diff_recur} as well as \eqref{SDYM-3}.
\end{theorem}

Thus, the master functions $\{\bs S^{(i,j)}\}$ for our purpose are well defined, so is $\bs V$, as a solution to
equation  \eqref{SDYM-3}.

\begin{theorem}\label{T-2-2}
Equation \eqref{SDYM-3} has the solution
\begin{align}\label{V}
    \bs V=\bs I_{\mathcal{N}}-\bs s^T(\bs C+\bs M)^{-1}\bs K^{-1}\bs r,
\end{align}
where the involved elements are subject to the two cases described in Theorem \ref{T-2-1}.
Moreover, the determinant $|\bs V|$ is a constant, and in particular,
$|\bs V|=1$ in the symmetric Sylvester equation case.
\end{theorem}

\begin{proof}
Recalling the  Weinstein--Aronszajn formula (see Theorem D.2 in \cite{HJN-2016})
\begin{equation}\label{WA}
|\bs I_N+\bs A \bs B|= |\bs I_{\mathcal{N}}+\bs B \bs A |,
\end{equation}
where $\bs A$ and $\bs B$ are $N\times \mathcal{N}$ and $\mathcal{N} \times N$ matrices,
from \eqref{V} and \eqref{Syl_eq} we have
\begin{align*}
        |\bs V|&=|\bs I_{\mathcal{N}}-\bs s^T(\bs C+\bs M)^{-1}\bs K^{-1}\bs r|\\
        &=|\bs I_N-(\bs C+\bs M)^{-1}\bs K^{-1}\bs r\bs s^T|\\
        &=|\bs I_N-(\bs C+\bs M)^{-1}\bs K^{-1}(\bs K\bs M-\bs M\bs L)| \\
        &=|\bs I_N-(\bs C+\bs M)^{-1}\bs K^{-1}(\bs K(\bs C+\bs M)-(\bs C+\bs M)\bs L-(\bs K\bs C-\bs C\bs L))| \\
        &=|(\bs C+\bs M)^{-1}\bs K^{-1}(\bs C+\bs M)\bs L+(\bs C+\bs M)^{-1}\bs K^{-1}(\bs K\bs C-\bs C\bs L)|.
    \end{align*}
In light of \eqref{KC=CL}, it then follows that
\begin{align}\label{det_V}
    |\bs V|=|(\bs C+\bs M)^{-1}\bs K^{-1}(\bs C+\bs M)\bs L|=\frac{|\bs L|}{|\bs K|}.
\end{align}
When $\bs K=\bs L$, we get $|\bs V|=1$.

\end{proof}

\begin{remark}\label{R-2-1}
    The Cauchy matrix approach used in \cite{LQYZ-SAPM-2022} is
    the symmetric Sylvester equation case with $\mathcal{M}=1$.
\end{remark}

\begin{remark}\label{R-2-2}
An alternative way to obtain equation \eqref{diff_recur} will be presented in Appendix \ref{APP-A}.
\end{remark}

\section{Asymmetric Sylvester formulation for the SDYM equation}\label{sec-3}

We have proved that equation \eqref{SDYM-3} admits a solution $\bs V$
that is well-defined through the Sylvester equations, as described in Theorem \ref{T-2-1}.
Note that at this moment $\bs V$ is not a solution of the SU$(\mathcal{N})$ SDYM equation \eqref{SDYM}
since the two equations,  \eqref{SDYM-3} and \eqref{SDYM}, are apparently different with a sign ``$-$''.
We need to elaborate the solutions of  equation \eqref{SDYM-3} by imposing extra constraints
on the coordinates $\mathbf{x}$ so that in the new coordinates system $\bs V$ can solve
the SU$(\mathcal{N})$ SDYM equation \eqref{SDYM}.

In this section, we investigate functions $\bs V$ defined in the asymmetric case.
We will first present explicit solutions $\bs r, \bs s$ and $\bs M$ of this case.
Then we will explore the constraints so that under which $\bs V$ solves the SDYM equation \eqref{SDYM}.
We will also examine the physical significance of $\bs V$,
including Hermitian property, determinant and positive-definiteness property of $\bs V$.

\subsection{Explicit solutions to the Sylvester equation \eqref{Syl_eq}}\label{sec-3-1}

In the asymmetric case, $\bs K$ and $\bs L$ do not share eigenvalues and $\bs C=\bs0$.
Thus $\bs V$ is written as
\begin{align}\label{V:C=0}
    \bs V=\bs I_{\mathcal{N}}-\bs s^T\bs M^{-1}\bs K^{-1}\bs r,
\end{align}
where $\{\bs M, \bs r, \bs s\}$ are solutions of system \eqref{Syl_eq} and \eqref{rs_move}.
Note that $\bs K, \bs L$ and any matrices that are similar to them
give rise to same $\bs S^{(i,j)}$ and consequently same $\bs V$, (see \cite{FZ-KP-2022,XZZ-2014,ZZ-2013}).
In this context, in the rest part of this section, we suppose
$\bs K$ and $\bs L$ are already of their canonical forms.
%Note that for $\mathcal{N }> 1$, the rank of $\bs r\bs s^T$ is larger than one.
To have an explicit form of solution $\bs M$ to the Sylvester equation \eqref{Syl_eq},
we first consider the following rank one case:\footnote{The right-hand side of \eqref{system-lem2a}
is a matrix with rank one.}
\begin{subequations}\label{system-lem2}
\begin{equation}\label{system-lem2a}
\bs K\bs M^{(i)}-\bs M^{(i)}\bs L=\bs r^{(i)}(\bs s^{(i)})^T,
\end{equation}
where $\bs M^{(i)}$ is the unknown $N\times N$ matrix to be determined,
$\bs r^{(i)}, \bs s^{(i)}$ and $a^{(i)}$ are components of $\bs r, \bs s$ and $\bs a$ (see \eqref{rs_exp}),
and satisfy
\begin{equation}\label{system-lem2b}
        \bs r^{(i)}_{x_n}=a^{(i)}\bs K^n\bs r^{(i)},
        ~~\bs s^{(i)}_{x_n}=-a^{(i)}(\bs L^T)^n\bs s^{(i)}.
    \end{equation}
\end{subequations}
This system allows explicit solutions (see \cite{FZ-KP-2022}).
To present these solutions, we list some special matrices and their properties.

Since $\bs K,\bs L$ are of their canonical forms,
we assume\footnote{$\bs K$ and $\bs L$ are diagonal when $n_i\equiv m_j \equiv 1$, for $i=1, 2, \cdots, p$
and $j=1,2,\cdots, q$.}
\begin{align}\label{KL}
    \bs K=\mr{diag}(\bs\Gamma_{n_1}(k_1),\bs\Gamma_{n_2}(k_2),\cdots,\bs\Gamma_{n_p}(k_p)),~~~~
    \bs L=\mr{diag}(\bs\Gamma_{m_1}(l_1),\bs\Gamma_{m_2}(l_2),\cdots,\bs\Gamma_{m_q}(l_q)),
\end{align}
where $\bs\Gamma_n(k)$ denotes the $n$-th order Jordan block
\begin{align}
    \bs\Gamma_{n}(k)=\begin{pmatrix}
        k & 0  & 0  & \cdots & 0 & 0 \\
        1   & k & 0 & \cdots & 0 & 0   \\
        0   &   1   &  k & \cdots & 0 & 0  \\
        \vdots    & \vdots  & \vdots  & \vdots & \vdots  & \vdots   \\
        0   & 0 & 0 & \cdots & 1 & k
    \end{pmatrix}_{n \times n},
\end{align}
the indices $\{n_i,m_j\}$ in \eqref{KL} satisfy $\sum_{i=1}^pn_i=\sum_{j=1}^qm_j=N$,
and $\{k_i\}$ and $\{l_j\}$ are eigenvalues of $\bs K$ and $\bs L$, respectively.
Then we introduce a $M\times M$  lower triangular Toeplitz matrix
\begin{align}\label{F}
\bs F_M=
        \begin{pmatrix}
            a_0 & 0 & 0 & \cdots & 0 \\
            a_1 & a_0 & 0 & \cdots & 0 \\
            a_2 & a_1 & a_0 & \cdots & 0 \\
            \vdots & \vdots & \vdots & \ddots & \vdots \\
            a_{M-1} & a_{M-2}
            & a_{M-3}   & \cdots & a_0
        \end{pmatrix}, ~~ a_j\in \mathbb{C}, ~ a_0\neq 0,
    \end{align}
and a $M\times M$  symmetric matrix
\begin{align}\label{H}
    \bs H_{M}=\begin{pmatrix}
       b_0 & b_1 & \cdots  &  b_{M-2}    & b_{M-1} \\
    b_1& b_2 & \cdots & b_{M-1} &0 \\
    \vdots &\vdots&\begin{sideways}$\ddots$\end{sideways} & \vdots & \vdots \\
    b_{M-2} & b_{M-1} & \cdots & 0 & 0    \\
     b_{M-1} & \cdots & 0 & \cdots & 0 \\
    \end{pmatrix}, ~~ b_j\in \mathbb{C}, ~ b_{M-1}\neq 0.
\end{align}
Note that we assume $a_0$ and $b_{M-1}$ are not zero so that $\bs F_M$ and $\bs H_M$ are invertible.
If there are $C^\infty$ functions $f(k)$ and $g(k)$ such that
\[a_j=\frac{\partial^{j}_{k}f(k)}{j!},~~
b_j=\frac{\partial^{j}_{k}g(k)}{j!},~~ (j=0,1,\cdots, M-1),\]
we say  $\bs F_{M}$ and  $\bs H_{M}$ are generated by
$f(k)$ and $g(k)$, and denote them by $\bs F_{M}[f(k)]$ and  $\bs H_{M}[g(k)]$, respectively.
Thus we have $\bs\Gamma_{n}(k)=\bs F_{n}[f(k)]$ where $f(k)=k$.
It can be verified that
\begin{equation*}
\bs H_{M}[g(k)] \bs F_{M}[f(k)]=\bs H_{M}[g(k)f(k)],
\end{equation*}
which means $\bs H_{M}[g(k)] \bs F_{M}[f(k)]$ is symmetric and
\begin{equation}\label{HF}
\bs H_{M}[g(k)] \bs F_{M}[f(k)]=\bs H_{M}[f(k)] \bs F_{M}[g(k)].
\end{equation}
In general, we have the following properties \cite{ZZSZ-RMP-2014}.
\begin{proposition}\label{P-A-1}
By $\mathcal{T}_M$ and $\mathcal{H}_M$ we denote the sets
composed by the $M\times M$ matrices of the form \eqref{F} and \eqref{H}, respectively.
Then we have:
\begin{itemize}
\item[(1)]{ $\mathcal{T}_M$ is an Abelian group.}
\item[(2)]{ All elements in $\mathcal{H}_M$ are symmetric, i.e. $\bs H=\bs H^T,~ \forall \bs H \in \mathcal{H}_M$.}
\item[(3)]
{$\bs H\bs F   \in \mathcal{H}_M,  ~ \forall \bs F\in \mathcal{T}_M,~\forall \bs H \in \mathcal{H}_M$.}
\item[(4)]{ $\bs H^{-1} \bs Q \in \mathcal{T}_M$, ~$\forall  \bs H, \bs Q \in \mathcal{H}_M$.}
\end{itemize}
\end{proposition}

Next, let
\begin{equation}\label{eE}
\e_n=(\underbrace{1,0,0,\cdots,0}_{\text{$n$-dimensional}}), ~~
\bs E_{n_1n_2\cdots n_p}=(\e_{n_1},\e_{n_2},\dots,\e_{n_p})^T
\end{equation}
and
\begin{equation}\label{G}
\bs G=(\bs G_{n_i,m_j}(k_{i},l_{j}))_{p\times q}
\end{equation}
be a $p\times q$ block matrix,
where each $\bs G_{n_i,m_j}(k_{i},l_{j})$ is a $n_i\times m_j$ matrix
in which the $(s,t)$-th element is defined by
\[
\frac{(-1)^{s+t}\tbinom{s+t-2}{s-1}}{(k_{i}-l_{j})^{s+t-1}},
~~\mathrm{for}~ 1\leq s \leq n_i,~ 1\leq t \leq m_j,
\]
where $\tbinom{n}{s}=\frac{n!}{s!(n-s)!}$.
Apart from the above $\bs G$,
we also introduce $\bs F^{(i)}$ and $\bs H^{(i)}$:
\begin{align}\label{FH}
    &\bs F^{(i)}=\mr{diag}(\bs F_{n_1}[\rho^{(i)}(k_1)],\dots,\bs F_{n_p}[\rho^{(i)}(k_p)]),
     ~~\bs H^{(i)}=\mr{diag}(\bs H_{m_1}[\sigma^{(i)}(l_1)],\dots,\bs H_{m_q}[\sigma^{(i)}(l_q)]),
\end{align}
where $\rho^{(i)}(k)$ and $\sigma^{(i)}(l)$ are the plane wave factors defined by
\begin{align}\label{rho-sij-APP}
    \rho^{(i)}(k)=\exp\left(a^{(i)}\sum_{n\in\mb Z}k^nx_n\right)\overset{\circ}{\rho}{}^{(i)}(k), ~~
    \sigma^{(i)}(l)=\exp\left(-a^{(i)}\sum_{n\in\mb Z}l^nx_n\right)\overset{\circ}{\sigma}{}^{(i)}(l),
\end{align}
where $\overset{\circ}{\rho}{}^{(i)}(k)$ and $\overset{\circ}{\sigma}{}^{(i)}(l)$
are functions of $k$ and $l$, respectively, and
$a^{(i)}$ are the diagonal elements of $\bs a$.

\begin{remark}\label{R-A-1}
We may introduce $\mathcal{T}_{n_1n_2\cdots n_p}$ and $\mathcal{H}_{m_1m_2\cdots m_q}$
to denote the sets respectively composed by the block diagonal matrices with the form of
$\bs F^{(i)}$ and $\bs H^{(i)}$ as given in \eqref{FH}.
Then, the properties in Proposition \ref{P-A-1} can be extended to
$\mathcal{T}_{n_1n_2\cdots n_p}$ and $\mathcal{H}_{m_1m_2\cdots m_q}$.
In other words, the properties (1)-(4) in Proposition \ref{P-A-1}
are still valid if replacing $\mathcal{T}_{M}$ and $\mathcal{H}_{M}$
with $\mathcal{T}_{n_1n_2\cdots n_p}$ and $\mathcal{H}_{n_1n_2\cdots n_p}$.
These properties are useful in presenting solutions as well as investigating
symmetric relation of $\{S^{(i,j)}\}$, (see Sec.\ref{sec-4-1}).
\end{remark}

With the above notations, general solutions to \eqref{system-lem2} are described as the following
(see \cite{FZ-KP-2022}).

\begin{lemma}\label{L-3-1}
System \eqref{system-lem2} with $\bs K, \bs L$ given in \eqref{KL}
has solutions of the following form
    \begin{align}
        \bs M^{(i)}=\bs F^{(i)}\bs G\bs H^{(i)},~~ \bs r^{(i)}=\bs F^{(i)}\bs E_{n_1n_2\cdots n_p},
        ~~ \bs s^{(i)}=\bs H^{(i)}\bs E_{m_1m_2\cdots m_q},
    \end{align}
where $\bs E_{n_1n_2\cdots n_p}, \bs G, \bs F^{(i)}, \bs H^{(i)}$
are defined as in \eqref{eE}, \eqref{G} and \eqref{FH}.
\end{lemma}

Noticing that the Sylvester equation \eqref{Syl_eq} can be written as
\begin{align*}
    \bs K\bs M-\bs M\bs L=\bs r^{(1)}(\bs s^{(1)})^T+\bs r^{(2)}(\bs s^{(2)})^T+\cdots
    +\bs r^{(\mathcal{N})}(\bs s^{(\mathcal{N})})^T,
\end{align*}
and it has a unique solution when $\bs K$ and $\bs L$
do not share eigenvalues, we immediately get its solution from Lemma \ref{L-3-1}. See below.

\begin{theorem}\label{T-3-1}
Assume $\bs K$ and $\bs L$ are in their canonical forms \eqref{KL} and do not share eigenvalues. Then
the Sylvester equation \eqref{Syl_eq} has a  solution
\begin{align*}
    \bs M=\sum_{i=1}^{\mathcal{N}}\bs M^{(i)}=\sum_{i=1}^{\mathcal{N}}\bs F^{(i)}\bs G\bs H^{(i)},
\end{align*}
where $\bs F^{(i)}, \bs H^{(i)}$ and $\bs G$ are given in \eqref{FH} and \eqref{G}.
\end{theorem}

\subsection{Reduction to \eqref{SDYM}}\label{sec-3-2}

Introduce constraints
\begin{align}\label{conjugate}
    \bs L^T=-(\bar{\bs K})^{-1}, ~~\bar {\bs a}=\bs a,
\end{align}
where bar stands for complex conjugate. We investigate the changes of solutions and coordinates
resulting from the above constraints.
Comparing the two equations in \eqref{rs_move}, we find that
\begin{align}\label{s-r}
\bs s= \bar{\bs K}^{-1}\bar{\bs r}
\end{align}
is a solution to the second equation  in \eqref{rs_move}. Consequently, the Sylvester equation \eqref{Syl_eq}
is written as
\begin{align}\label{Syl_eq-2}
    \bs K\bs M\bs K^\dagger+\bs M=\bs r\bs r^\dagger,
\end{align}
where $\bs K^\dagger=\bar{\bs K}^T$.
Noticing that $\bs M^\dagger$ is a solution of \eqref{Syl_eq-2} as well
and the above equation has a unique solution, we immediately have
\[\bs M=\bs M^\dagger,\]
i.e., $\bs M$ is a Hermitian matrix.

Back to the first equation in \eqref{rs_move}, i.e. $\bs r_{x_n}=\bs K^n\bs r\bs a$.
Multiplying $\bar{\bs K}^{-1}$ on its complex conjugate form, in light of \eqref{conjugate} and \eqref{s-r}, we have
\[\bar{\bs K}^{-1}\bar{\bs r}_{\bar{x}_n}
=\bar{\bs K}^{-1} \bar{\bs K}^{n}\bar{\bs r}\bar{\bs a}
=\bar{\bs K}^{n} \bs s \bs a
=(-\bs L^T)^{-n} \bs s \bs a
=(-1)^n(\bs L^T)^{-n}\bs s\bs a,
\]
which gives rise to
\begin{equation}\label{s}
\bs s_{\bar{x}_n}=(-1)^{n+1}\bs s_{x_{-n}}.
\end{equation}
In a similar manner, if we start from the second equation in \eqref{rs_move},
i.e. $\bs s_{x_n}=-(\bs L^T)^n\bs s\bs a$, using \eqref{conjugate} and \eqref{s-r},
we can find that
\begin{equation}\label{r}
\bs r_{\bar{x}_n}=(-1)^{n+1}\bs r_{x_{-n}}.
\end{equation}
Thus, \eqref{s} and \eqref{r} together, indicate the consistent relation of coordinates
\begin{equation}\label{x}
{x}_n=(-1)^{n+1} \bar{x}_{-n},~~ n\in \mathbb{Z},
\end{equation}
and it then follows from \eqref{rho-sij-APP} that
\begin{equation}\label{3.19}
 \sigma^{(i)}(-1/\bar k)=\mu(k)\, \overline{\rho^{(i)}(k)},~~~
 \mu(k)= \overset{\circ}{\sigma}{}^{(i)}(-1/\bar k)/\overset{\circ}{\rho}{}^{(i)}(k).
\end{equation}

For the SDYM equation \eqref{SDYM}, by introducing\footnote{The case $n=0$ yields a 2-dimensional equation.
When $n=0$, relation \eqref{x}
yields $x_0=-\bar{x}_0$, which means $\xi_0=0$, i.e. $y_0=x_0=\mathrm{i}\eta_0$.
Consequently,   for $n=0$, \eqref{SDYM-4} reads
$ (\bs V_{y_{1}}\bs V^{-1})_{\bar{y}_{1}}=(\bs V_{x_0}\,\bs V^{-1})_{x_0}$.
Noticing that \eqref{Sij_moves2} yields $\bs V_{x_0}=[\bs V, \bs a]$, it then follows that
$ (\bs V_{y_{1}}\bs V^{-1})_{\bar{y}_{1}}=[\bs V\bs a\bs V^{-1},\bs a]$,
which is a 2-dimensional equation.}
\begin{align}\label{coor}
y_n\doteq x_n=\xi_n+\I \eta_n, ~~
\bar{y}_n\doteq (-1)^{n+1}x_{-n}=\xi_n-\I \eta_n,~~
n=1,2,\cdots,
\end{align}
where $\I^2=-1$, $\xi_n, \eta_n \in \mathbb{R}$,
and setting $m=-(n+1)$ in \eqref{SDYM-3}, we have
\begin{align}\label{SDYM-4}
    (\bs V_{y_{n+1}}\bs V^{-1})_{\bar{y}_{n+1}}+(\bs V_{\bar{y}_n}\,\bs V^{-1})_{y_n}=0,
\end{align}
which agrees with equation \eqref{SDYM}.
Thus, we come to the following theorem.

\begin{theorem}\label{T-3-2}
With the constraint \eqref{conjugate} and coordinates $\{y_n\}$ defined in \eqref{coor},
the function
\begin{align}\label{V:asym}
    \bs V=\bs I_{\mathcal{N}}-\bs r^\dagger(\bs K^{\dagger})^{-1}\bs M^{-1}\bs K^{-1}\bs r
\end{align}
is a solution to equation \eqref{SDYM-4}
and $J=\bs V$ is a solution to the SDYM equation \eqref{SDYM} with
$y=y_n, z=\bar{y}_{n+1}$.
$\bs V$ is a Hermitian matrix since $\bs M=\bs M^\dagger$.
In addition, in light of \eqref{det_V} and \eqref{conjugate}, we have
\begin{align*}
|\bs V|=\frac{|\bs L|}{|\bs K|}=\frac{(-1)^N}{|\bs K||\bar{\bs K}|}.
\end{align*}
\end{theorem}

\vskip 3pt
\begin{remark}\label{R-3-1}
Function $\bs V$ given in \eqref{V:asym} provides a solution to the SDYM equation \eqref{SDYM}.
Considering physical significance, $\bs V$ should be Hermitian, positive-definite and $|\bs V|=1$.
When $\mathcal{N}$ is odd or $N$ is even,
one can always normalize $\bs V$ by
\begin{equation*}%\label{V'}
    \bs V'=(\sqrt[\mathcal{N}]{|\bs V|})^{-1}\bs V
\end{equation*}
such that $\bs V'$ is still a solution to \eqref{SDYM-4},  $\bs V'$ is Hermitian  and  $|\bs V'|=1$.
However, the positive-definiteness of $\bs V$ is uncertain.
\end{remark}

We will analyze the positive-definiteness of $\bs V$ with examples in the following.

\subsection{Positive definiteness}\label{sec-3-3}

It is hard to discuss the positive-definiteness of $\bs V$ for arbitrary $\mathcal{N}$ and $N$.
In the following we consider two special cases.

\subsubsection{One-soliton solution of the SU($\mathcal{N}$) SDYM equation}\label{sec-3-3-1}

We consider the case of $N=1$ while $\mathcal{N}$ being arbitrary.
In this case, $\bs K=k_1$
%\begin{align*}
%    \bs K=k_1
%\end{align*}
where $k_1\in \mathbb{C}$ and $k_1\neq 0$;
$\bs r$ is a $1\times \mathcal{N}$ matrix
$ \bs r=(\rho_1^{(1)}, \rho_1^{(2)},\cdots,\rho_1^{(\mathcal{N})})$,
%\begin{align*}
%   \bs r=(\rho_1^{(1)}, \rho_1^{(2)},\cdots,\rho_1^{(\mathcal{N})}),
%\end{align*}
where (for $j=1$)
\begin{equation}\label{rho-ji}
    \rho_j^{(i)}= \rho^{(i)}(k_j)=
    \exp\biggl[ a^{(i)}\sum_{m=n}^{n+1}\big((k_j^m+(-1)^{m+1}k_j^{-m})\xi_m
    +\mr i(k_j^m-(-1)^{m+1}k_j^{-m})\eta_m\big)\biggr]\overset{\circ}{\rho}_j\!{}^{(i)},
\end{equation}
$a^{(i)}\in \mathbb{R}$ and $\overset{\circ}{\rho}_j\!{}^{(i)}\in \mathbb{C}$;
the matrix $\bs M$ in this case is a scalar function
\begin{align*}
    \bs M=\frac{\bs r\bs r^\dagger}{|k_1|^2+1}=\frac{\sum_{i=1}^{\mathcal{N}}|\rho_1^{(i)}|^2}{|k_1|^2+1}.
\end{align*}
One can also write the plane wave factor $ \rho_j^{(i)}$ defined in \eqref{rho-ji} as
\begin{equation*}%\label{rho-ji-L}
\rho_j^{(i)}= \exp\bigl[a^{(i)}\mathfrak L_j(\bs x)\bigr] \overset{\circ}{\rho}_j\!{}^{(i)},
\end{equation*}
where
\begin{subequations}\label{Lcx}
\begin{align}
   &\mathfrak L_j(\bs x)
    =\bs c_j\bs x
    =\sum_{m=n}^{n+1}(k_j^m+(-1)^{m+1}k_j^{-m})\xi_m+\mr i(k_j^m-(-1)^{m+1}k_j^{-m})\eta_m,
    \label{L(x)}\\
    &\bs c_j= \left(k_j^n-\frac{(-1)^{n}}{k_j^{n}}, ~k_j^{n+1}+\frac{(-1)^{n}}{k_j^{n+1}}, ~
    \mr i k_j^n+\frac{\mr i(-1)^{n} }{k_j^{n}}, ~\mr i k_j^{n+1}-\frac{\mr i  (-1)^{n}}{k_j^{n+1}}\right),
    \label{cj}\\
    &\bs x=(\xi_n,\xi_{n+1},\eta_n,\eta_{n+1})^T.
\end{align}
\end{subequations}
By definition \eqref{UV} (see also \eqref{V:C=0}) we have
\begin{equation*}%\label{V-ij}
 \bs V=(v_{ij})_{\mathcal{N}\times \mathcal{N}}
 =\bs I_{\mathcal{N}}-\bs S^{(-1,0)}
 =\bs I_{\mathcal{N}}-(s^{(-1,0)}_{ij})_{\mathcal{N}\times \mathcal{N}}.
 \end{equation*}
To determine the positive-definiteness of $\bs V$, we investigate its leading principal minors, which  are denoted as
\begin{align*}
    D_1=v_{11},~~
    D_2=\begin{vmatrix}
        v_{11} & v_{12} \\
        v_{21} & v_{22}
    \end{vmatrix},~~
    D_3=\begin{vmatrix}
        v_{11} & v_{12} & v_{13} \\
        v_{21} & v_{22} & v_{23} \\
        v_{31} & v_{32} & v_{33}
    \end{vmatrix},~~
    \cdots,~~
    D_{\mathcal{N}}=|\bs V|.
\end{align*}
Noticing that in this case\footnote{This indicates all elements in $\bs S^{(-1,0)}$
are nonsingular globally for all $(\xi_n,\eta_n,\xi_{n+1},\eta_{n+1})\in \mathbb{R}^4$,
and so are the elements in $\bs V$.}
\begin{align*}
\bs S^{(-1,0)}=\bs s^T\bs M^{-1}\bs K^{-1}\bs r
                    =\bs r^\dagger(\bs K^\dagger)^{-1}\bs M^{-1}\bs K^{-1}\bs r
                    =\biggl(1+\frac{1}{|k_1|^2}\biggr)
                    \frac{\bs r^\dagger\bs r }{\sum_{i=1}^{\mathcal{N}}|\rho_1^{(i)}|^2},
\end{align*}
we make use of the Weinstein--Aronszajn formula (see \eqref{WA})
which in this case gives
\begin{equation*}%\label{WA}
|I_{\mathcal{N}}- \varepsilon \bs r^\dagger\bs r|
=1-\varepsilon \bs r\bs r^\dagger,
\end{equation*}
where $\varepsilon=(1+\frac{1}{|k_1|^2})\frac{1}{\sum_{i=1}^{\mathcal{N}}|\rho_1^{(i)}|^2}$.
Thus, for the  leading principal minors we have
\begin{align*}
  D_{l}  =1-\biggl(1+\frac{1}{|k_1|^2}\biggr)\frac{\sum_{i=1}^{l}|\rho_1^{(i)}|^2}
  {\sum_{i=1}^{\mathcal{N}}|\rho_1^{(i)}|^2},~
    ~~ l =1,\cdots, \mathcal{N},
\end{align*}
Clearly, the sequences satisfy
\begin{align*}
    D_1>D_2>\cdots>D_{\mathcal{N}-1}>D_{\mathcal{N}}=-\frac{1}{|k_1|^2}.
\end{align*}
Such a relation holds globally for all
$(\xi_n,\eta_n,\xi_{n+1},\eta_{n+1})\in \mathbb{R}^4$.
This fact indicates that, when $N=1$ and $\mathcal{N}\geq 2$, in the asymmetric case,
the matrix $\bs V$ is neither positive definite  nor negative definite
in any domain in $\mathbb{R}^4$.
As a conclusion, we have the following.

\begin{proposition}\label{P-5-1}
In the asymmetric case, when $\mathcal{N}\geq 2$,
the  one-soliton solution (i.e. $N=1$) $\bs V$   is neither positive definite nor negative definite,
no matter $\mathcal{N}$ is odd or even.
\end{proposition}

\subsubsection{Two-soliton  solution of the  SU(2) SDYM equation}\label{sec-3-3-2}

Now we investigate the two-soliton solution of the  SU(2) SDYM equation,
i.e. $\mathcal{N}=2$ and $N=2$,
where $\bs K$ and $\bs r$ are given by
\begin{align*}
    \bs K=
    \begin{pmatrix}
        k_1 & \\
         & k_2
    \end{pmatrix},
    ~~~~\bs r=(\bs r^{(1)},\bs r^{(2)})=
    \begin{pmatrix}
        \rho_1^{(1)} &  \rho_1^{(2)} \\
        \rho_2^{(1)} &  \rho_2^{(2)}
    \end{pmatrix},
\end{align*}
and $\rho_j^{(i)}$  are defined as in \eqref{rho-ji}.
Note here that $\bs r^{(j)}=( \rho_1^{(j)},  \rho_2^{(j)})^T$, for $j=1,2$.
The dressed Cauchy matrix $\bs M$ is (referring to Lemma \ref{L-3-1} and Theorem \ref{T-3-1}
where $p=q=2, n_1=n_2=m_1=m_2=1$)
\[  \bs M=\bs F^{(1)}\bs G\bs H^{(1)}+\bs F^{(2)}\bs G\bs H^{(2)},
\]
where
\begin{align*}
\bs G=\begin{pmatrix}
        \frac{1}{k_1+1/{\bar k_1}} & \frac{1}{k_1+1/{\bar k_2}} \\
         \frac{1}{k_2+1/{\bar k_1}} & \frac{1}{k_2+1/{\bar k_2}}
    \end{pmatrix},\;\;
\bs F^{(i)}=\begin{pmatrix}
        \rho_1^{(i)} & \\
         & \rho_2^{(i)}
    \end{pmatrix},\;\;
\bs H^{(i)}=\begin{pmatrix}
        &\bar \rho_1^{(i)}/{\bar k_1^{}} \\
         \bar\rho_2^{(i)}/{\bar k_2} &
    \end{pmatrix}, ~~~  i=1,2,
\end{align*}
i.e.
\begin{align*}
    \bs M=
    \begin{pmatrix}
        \frac{\sum_{i=1}^{2}|\rho_1^{(i)}|^2}{|k_1|^2+1} &
        \frac{\sum_{i=1}^{2}\rho_1^{(i)}\bar\rho_2^{(i)}}{k_1\bar k_2+1} \\
        \frac{\sum_{i=1}^{2}\rho_2^{(i)}\bar\rho_1^{(i)}}{k_2\bar k_1+1} &
        \frac{\sum_{i=1}^{2}|\rho_2^{(i)}|^2}{|k_2|^2+1}
    \end{pmatrix}.
\end{align*}
Here we have made use of relation \eqref{3.19}
and already taken $\mu(k)=1/\bar k$  so that $\bs M$ is a Hermitian matrix.
$\bs V$ is given by \eqref{V:asym}.
Since, according to Theorem \ref{T-3-2},  $|\bs V|=|k_1k_2|^{-2}$ is positive,
the positive-definiteness of $\bs V$ is therefore determined by the sign of $v_{11}$,
which is
%\begin{equations}
\begin{align}
    v_{11}&=1-s_{11}^{(-1,0)}=1-(\bs r^{(1)})^\dagger(\bs K^\dagger)^{-1}\bs M^{-1}\bs K^{-1}\bs r^{(1)}
    \nonumber \\
    &=1-\frac{1}{|k_1k_2|^2|\bs M|}
    \bigg(\frac{|k_2|^2|\rho_1^{(1)}|^2}{|k_2|^2+1}\sum_{i=1}^{2}|\rho_2^{(i)}|^2
        +\frac{|k_1|^2|\rho_2^{(1)}|^2}{|k_1|^2+1}\sum_{i=1}^{2}|\rho_1^{(i)}|^2  \nonumber\\
    &~~~~~-2\mr{Re}\Bigl[\frac{k_1\bar k_2}{k_1\bar k_2+1}
    \bar\rho_1^{(1)}\rho_2^{(1)}\sum_{i=1}^{2}\rho_1^{(i)}\bar\rho_2^{(i)}\Bigr]
    \bigg). \label{v11KP}
\end{align}

In the following we are going to develop an approach to get some localized domains in $\mathbb{R}^4$
where $v_{11}$ is positive.
The idea is described as follows.
We first determine the singular points of $v_{11}$.
In principle, $v_{11}$ is continuous in $\mathbb{R}^4$ except on those singular points.
We can find a point $\bs x_0\in \mathbb{R}^4$  on which $v_{11}$ is positive,
and then utilizing the local property of a continuous function,
we can have a neighbourhood  of $\bs x_0$ in which $v_{11}$ takes positive value.

Let us elaborate this approach below.
The singularity of $v_{11}$ takes place when $|\bs M|=0$,
i.e.
\begin{align*}%\label{DetM=0}
    \frac{(\sum_{i=1}^{2}|\rho_1^{(i)}|^2)(\sum_{i=1}^{2}|\rho_2^{(i)}|^2)}
    {(|k_1|^2+1)(|k_2|^2+1)}
    -\frac{|\sum_{i=1}^{2}\rho_1^{(i)}\bar\rho_2^{(i)}|^2}{|k_1\bar k_2+1|^2}=0.
\end{align*}
We write it as the following simpler form
\begin{equation}\label{w-k}
    \frac{|w_1\bar w_2+1|^2}{(|w_1|^2+1)(|w_2|^2+1)}
=\frac{|k_1\bar k_2+1|^2}{(|k_1|^2+1)(|k_2|^2+1)},
\end{equation}
where
\begin{equation}\label{w-x}
 w_j=\rho_j^{(1)}/\rho_j^{(2)}=
    \overset{\circ}{w}_j \exp\bigl[(a^{(1)}-a^{(2)})\mathfrak L_j(\bs x)\bigr], ~~~
    \overset{\circ}{w}_j= \overset{\circ}{\rho}_j\!{}^{(1)}/\overset{\circ}{\rho}_j\!{}^{(2)},
\end{equation}
and $\mathfrak L_j(\bs x)$ are defined as in \eqref{Lcx}.
Denoting
$w_j=a_j+\mr i b_j, ~ a_j=\mathrm{Re}[w_j],~ b_j=\mathrm{Im}[w_j]$,
equation \eqref{w-k} gives rise to
\begin{equation}\label{w-k-ab}
\frac{(a_1a_2+b_1b_2+1)^2+(a_2b_1-a_1b_2)^2}{(a_1^2+b_1^2+1)(a_2^2+b_2^2+1)}
=\frac{|k_1\bar k_2+1|^2}{(|k_1|^2+1)(|k_2|^2+1)}.
\end{equation}
Noticing that both $a_j$ and $b_j$ are functions of $\bs x$, in principle,
for given $(k_1,k_2)$, equation \eqref{w-k-ab} determines an implicit function
(e.g. $\xi_{n+1}=\mathcal{G}(\xi_n, \eta_n, \eta_{n+1})$),
i.e. a surface in $\mathbb{R}^4$ on which $v_{11}$ is singular.
Note that it is easy to see $(w_1, w_2)=\pm (k_1, k_2)$ and
$(w_1, w_2)=\pm (\bar k_1, \bar k_2)$ are the points on the surface \eqref{w-k-ab}.

For the sake of convenience, our investigation on $v_{11}$ will be implemented in terms of $\{w_j, \bar w_j\}$
rather than $\bs x$.
Equation \eqref{w-x} defines $w_j$ as functions of $\bs x$ and $\{k_j\}$.
The other way around,  $\bs x$ can be expressed via $\{w_j\}$  and $\{k_j\}$.

\begin{proposition}\label{P-3-1}
For given $\{w_j\}$ and $\{k_j\}$, from \eqref{w-x}, the corresponding coordinate $\bs x$
can be recovered via the formula
\begin{align}\label{3.33}
    \bs x =
    \begin{pmatrix}
        \mr{Re}[\bs c_1 ] \\
        \mr{Im}[\bs c_1 ] \\
        \mr{Re}[\bs c_2 ] \\
        \mr{Im}[\bs c_2 ]
    \end{pmatrix}^{-1}
    \begin{pmatrix}
        \mr{Re}[\mathcal W_1] \\
        \mr{Im}[\mathcal W_1] \\
        \mr{Re}[\mathcal W_2] \\
        \mr{Im}[\mathcal W_2]
    \end{pmatrix},
\end{align}
where $\bs c_j$ is defined in \eqref{cj} and
\[\mathcal{W}_j = \frac{\ln(w_j/\overset{\circ}{w}_j)+2s \pi\mr i}{a^{(1)}-a^{(2)}}, ~~  s\in\mb Z.\]
\end{proposition}

Before we proceed, we look at the following real-valued function,
\begin{equation}\label{w-F}
 \mathcal{F}(w_1, w_2)=   \frac{|w_1\bar w_2+1|^2}{(|w_1|^2+1)(|w_2|^2+1)},
\end{equation}
with which  \eqref{w-k}  is written as $ \mathcal{F}(w_1, w_2)= \mathcal{F}(k_1, k_2)$.
Defining
\begin{equation}\label{w-alb}
w_j=\exp(\alpha_j+ \mathrm{i}\beta_j),~~ \alpha_j, \beta_j \in \mathbb{R}[\bs x],
\end{equation}
we have
\begin{equation*}%\label{w-F-alb}
 \mathcal{F}(w_1, w_2)
 =1-\frac{\exp(2\alpha_1)+\exp(2\alpha_2)-2\exp(\alpha_1+\alpha_2)
 \cos(\beta_1-\beta_2)}{(\exp(2\alpha_1)+1)(\exp(2\alpha_2)+1)},
\end{equation*}
which indicates the following asymptotic property.

\begin{proposition}\label{P-3-2}
For the function $\mathcal{F}(w_1, w_2)$ defined in \eqref{w-F}, we have
\begin{equation*}
\mathcal{F}(w_1, w_2) \sim \left\{
\begin{array}{lll}
1, && (|w_1|, |w_2|)\to (\infty, \infty) \mr{~or~}  (0,0),\\
\frac{1}{\exp(2\alpha_1)+1}, && |w_1|~\mathrm{finite}, ~ |w_2| \to 0,\\
\frac{\exp(2\alpha_1)}{\exp(2\alpha_1)+1}, && |w_1|~\mathrm{finite}, ~ |w_2| \to \infty.
\end{array}\right.
\end{equation*}
\end{proposition}

Now let us proceed to investigate $v_{11}$, to look at its asymptotic
behavior with respect to $(|w_1|, |w_2|)$.
Rewrite $v_{11}$ in terms of $w_j$ as the following
\begin{equation}\label{v11-fg}
v_{11}=\frac{1+\frac{|k_1\bar k_2+1|^2}{|k_1-k_2|^2} \Bigl(\bigl| \frac{w_1}{k_1}\bigr|^2
+\bigl|\frac{w_2}{k_2}\bigr|^2\Bigr)
+ \bigl|\frac{w_1w_2}{k_1k_2}\bigr|^2
-2\frac{(1+|k_1|^2)(1+|k_2|^2}{|k_1-k_2|^2}\mathrm{Re}\bigl[\frac{w_1\bar w_2}{k_1 \bar k_2}\bigr]}
{1-\frac{|k_1\bar k_2+1|^2}{|k_1-k_2|^2} (|w_1|^2+|w_2|^2)
+ |w_1w_2|^2
+2\frac{(1+|k_1|^2)(1+|k_2|^2}{|k_1-k_2|^2}\mathrm{Re}[w_1\bar w_2]}.
\end{equation}
Note that we have the following relation
\begin{equation}\label{Re-ww}
\mathrm{Re}[w_1 w_2] \to 0 ~~~~ \mathrm{when} ~~ (|w_1|, |w_2|) \to (0,0),
\end{equation}
which is easy to  obtain from
$\mathrm{Re}[w_1 w_2]=\exp(\alpha_1+\alpha_2) \cos(\beta_1+\beta_2)$
in light of the expression \eqref{w-alb}.
Then, by computation we obtain the asymptotic property of $v_{11}$. 

\begin{proposition}\label{P-3-3}
$v_{11}$ has the asymptotic property in terms of $|w_j|$, as given in Table \ref{Tab-2}.
\begin{table}[ht]
\begin{center}
\begin{tabular}{|c|c|c|} \hline
 limits & $|w_1|\rightarrow0$ & $|w_1|\rightarrow\infty$   \\ \hline
$|w_2|\rightarrow0$ & $v_{11}\to 1 $ & $v_{11}\to -|k_1|^{-2}$   \\ \hline
$|w_2|\rightarrow\infty$ & $v_{11}\to -|k_2|^{-2}$ 	 & $v_{11}\to |k_1k_2|^{-2} $  \\ \hline
\end{tabular}
\end{center}
\caption{Asymptotic property of $v_{11}$}\label{Tab-2}
\end{table}
\end{proposition}

Now we  sketch the steps by which we can get a domain in $\mathbb{R}^4$ where $\bs V$ is definite-positive.
First, instructed by Proposition \ref{P-3-3} we choose $w_1$ and $w_2$
of which both $|w_1|$ and $|w_2|$ are small enough (or large enough).
In this case, in light of Proposition \ref{P-3-2}, $\mathcal{F}(w_1,w_2)$ has a value  close to 1.
Next, again  in light of Proposition \ref{P-3-2}, we take $(k_1, k_2)$
such that $\mathcal{F}(k_1,k_2)$ is far from $\mathcal{F}(w_1,w_2)$.
In the final step we recover a point $\bs x$ from $(w_1, w_2, k_1, k_2)$ using formula \eqref{3.33} in
Proposition \ref{P-3-1}.
Thus there is a neighbourhood of $\bs x$ where $v_{11}$ is positive and  $\bs V$ is definite-positive.

To compare with the two-soliton solution obtained
from the symmetric Sylvester formulation (see
Sec.\ref{sec-4-3-2}), below we give the explicit expressions of $v_{12}$ and $v_{22}$,
which are formulated in \eqref{V:asym} and can be expressed as
\begin{equation}\label{v21-fg}
v_{12}=\frac{\frac{(1+\bar k_1 k_2)(1+|k_2|^2)\bar w_2 |w_1|^2}{(k_1-k_2)\bar k_1 |k_2|^2}
-\frac{(1+ k_1\bar k_2)(1+|k_1|^2)\bar w_1 |w_2|^2}{(k_1-k_2)\bar k_2 |k_1|^2}
-\frac{(1+|k_1|^2)(1+\bar k_1 k_2)\bar w_1}{(\bar k_1-\bar k_2)k_2 |k_1|^2}
+\frac{(1+|k_2|^2)(1+ k_1\bar k_2)\bar w_2}{(\bar k_1-\bar k_2)k_1 |k_2|^2}}
{1-\frac{|k_1\bar k_2+1|^2}{|k_1-k_2|^2} (|w_1|^2+|w_2|^2)
+ |w_1w_2|^2
+2\frac{(1+|k_1|^2)(1+|k_2|^2}{|k_1-k_2|^2}\mathrm{Re}[w_1\bar w_2]}
\end{equation}
and
\begin{equation}\label{v22-fg}
v_{22}=\frac{1+\frac{|k_1\bar k_2+1|^2}{|k_1-k_2|^2} (|k_1 w_1|^2+|k_2 w_2|^2)
+ |k_1k_2w_1w_2|^2
-2\frac{(1+|k_1|^2)(1+|k_2|^2}{|k_1-k_2|^2}\mathrm{Re}[k_1\bar k_2 w_1\bar w_2]}
{|k_1k_2|^2\Bigl (1-\frac{|k_1\bar k_2+1|^2}{|k_1-k_2|^2} (|w_1|^2+|w_2|^2)
+ |w_1w_2|^2
+2\frac{(1+|k_1|^2)(1+|k_2|^2}{|k_1-k_2|^2}\mathrm{Re}[w_1\bar w_2]\Bigr)}.
\end{equation}
We also present  their deformations by redefining $\overset{\circ}{w}_j$ such that
\begin{equation}
w_1= \frac{ k_1-k_2 }{1+k_1\bar k_2} \widetilde{w}_1,~~
w_2= -\frac{ k_1-k_2 }{1+k_2\bar k_1} \widetilde{w}_2.
\end{equation}
In terms of $\widetilde{w}_j$, $v_{11}$, $v_{21}$ and $v_{22}$ are respectively  written as
\begin{equation}\label{v11-wt}
v_{11}=\frac{1+\bigl|\frac{\widetilde{w}_1}{k_1}\bigr |^2+\bigl|\frac{\widetilde{w}_2}{k_2}\bigr |^2
+ \frac{|k_1-k_2|^4}{|1+k_1\bar k_2|^4}\bigl|\frac{\widetilde{w}_1\widetilde{w}_2}{k_1k_2}\bigr |^2
+2 (1+|k_1|^2)(1+|k_2|^2 )
\mathrm{Re}\Bigl[\frac{\widetilde{w}_1\bar{\widetilde{w}}_2}{k_1\bar k_2(1+k_1\bar k_2)^2}\Bigr] }
{1-| \widetilde{w}_1|^2-| \widetilde{w}_2|^2
+ \frac{|k_1-k_2|^4}{|1+k_1\bar k_2|^4}|\widetilde{w}_1\widetilde{w}_2|^2
-2 (1+|k_1|^2)(1+|k_2|^2)
\mathrm{Re}\Bigl[\frac{\widetilde{w}_1 \bar{\widetilde{w}}_2}{(1+k_1\bar k_2)^2}\Bigr]},
\end{equation}
\begin{equation}\label{v12-wt}
v_{12}=\frac{- (\bar k_1-\bar k_2)^2
\Bigl(\frac{k_1 (1+|k_2|^2)\bar{\widetilde{w}}_2 |\widetilde{w}_1|^2}{(1+\bar k_1 k_2)^2}
+\frac{k_2(1+|k_1|^2) \bar{\widetilde{w}}_1 | \widetilde{w}_2|^2}{(1+ k_1\bar k_2)^2} \Bigr)
- \bar k_2(1+ |k_1|^2)  \bar{\widetilde{w}}_1
-\bar k_1(1+ |k_2|^2)  \bar{\widetilde{w}}_2 }
{|k_1k_2|^2 \Bigl(1-| \widetilde{w}_1|^2-| \widetilde{w}_2|^2
+ \frac{|k_1-k_2|^4}{|1+k_1\bar k_2|^4}|\widetilde{w}_1\widetilde{w}_2|^2
-2 (1+|k_1|^2)(1+|k_2|^2)
\mathrm{Re}\Bigl[\frac{\widetilde{w}_1 \bar{\widetilde{w}}_2}{(1+k_1\bar k_2)^2}\Bigr]\Bigr)}
\end{equation}
and
\begin{equation}\label{v22-wt}
v_{22}=\frac{1+|k_1 \widetilde{w}_1|^2+|k_2 \widetilde{w}_2|^2
+ \frac{|k_1-k_2|^4}{|1+k_1\bar k_2|^4}|k_1k_2\widetilde{w}_1\widetilde{w}_2|^2
+2 (1+|k_1|^2)(1+|k_2|^2 )
\mathrm{Re}\Bigl[\frac{k_1\bar k_2 \widetilde{w}_1\bar{\widetilde{w}}_2}
{(1+k_1\bar k_2)^2}\Bigr] }
{|k_1k_2|^2 \Bigl(1-| \widetilde{w}_1|^2-| \widetilde{w}_2|^2
+ \frac{|k_1-k_2|^4}{|1+k_1\bar k_2|^4}|\widetilde{w}_1\widetilde{w}_2|^2
-2 (1+|k_1|^2)(1+|k_2|^2)
\mathrm{Re}\Bigl[\frac{\widetilde{w}_1 \bar{\widetilde{w}}_2}{(1+k_1\bar k_2)^2}\Bigr]\Bigr)}.
\end{equation}
In addition, we note that one can always normalize the real coefficient $(a^{(1)}-a^{(2)})$ in \eqref{w-x}
to be $1$ since the SDYM equation \eqref{SDYM} is invariant with real scaling transformation
$(y, z) \rightarrow (ay, az)$ where $a\in \mathbb{R}$.

\section{Symmetric Sylvester formulation for the SDYM equation}\label{sec-4}

In this section we derive the SDYM equation \eqref{SDYM} from \eqref{SDYM-3}
where $\bs V$ is defined by \eqref{V} in the symmetric Sylvester equation case.
In this case $\bs K=\bs L$ and $\bs K \bs C=\bs C \bs K$.

\subsection{Explicit solutions to the Sylvester equations \eqref{symm-case}}\label{sec-4-1}

In light of assumption \eqref{KAM}, we are able to obtain solutions to the Sylvester equations \eqref{symm-case}
where $\bs r_j$ and $\bs s_j$ are defined by
\eqref{DR-symm}.
Both equations in \eqref{symm-case} are in the form of \eqref{Syl_eq}.
Thus we can simply use the results in Theorem \ref{T-3-1}  to present solutions to \eqref{symm-case}.
For the notations involved  in the following theorem, one can refer to Sec.\ref{sec-3-1}.

\begin{theorem}\label{T-4-1}
We assume $\bs K_1$ and $\bs K_2$ are of their canonical forms, i.e.
\begin{align}\label{K1K2}
    \bs K_1=\mr{diag}(\bs\Gamma_{n_1}(k_1),\bs\Gamma_{n_2}(k_2),\cdots,\bs\Gamma_{n_p}(k_p)),~~
    \bs K_2=\mr{diag}(\bs\Gamma_{m_1}(l_1),\bs\Gamma_{m_2}(l_2),\cdots,\bs\Gamma_{m_q}(l_q)),
\end{align}
where $\sum_{i=1}^{p}n_i=\sum_{i=1}^{q}m_i=M$,
and $\bs K_1$ and $\bs K_2$ do not share eigenvalues.
Note that  $\bs K_1 \in \mathcal{T}_{n_1n_2\cdots n_p}$ and
$\bs K_2 \in \mathcal{T}_{m_1m_2\cdots m_q}$.
Take $\bs C_1 \in \mathcal{T}_{n_1n_2\cdots n_p}$ and
$\bs C_2 \in \mathcal{T}_{m_1m_2\cdots m_q}$, i.e.
\begin{equation*}%\label{C1C2}
    \bs C_1=\mr{diag}(\bs C_{1,n_1},\bs C_{1,n_2},\cdots,\bs C_{1,n_p}),~~
    \bs C_2=\mr{diag}(\bs C_{2,m_1},\bs C_{2,m_2},\cdots,\bs C_{2,m_q}),
\end{equation*}
where $\bs C_{j,n} \in \mathcal{T}_n$   such that
$\bs K_j \bs C_j=\bs C_j \bs K_j$ for $j=1,2$.
There exist $\bs F_j^{(i)}, \bs H_j^{(i)}, \bs G_j$ and $\bs E_j$ such that
\begin{subequations}
\begin{align}
& \bs M_1^{(i)}=  \bs F_1^{(i)}\bs G_1\bs H_2^{(i)}, ~~
\bs M_2^{(i)}=  \bs F_2^{(i)}\bs G_2\bs H_1^{(i)}, ~~ \bs G_2=-\bs G_1^T, \\
&\bs r_j^{(i)}=\bs F_j^{(i)}\bs E_j, ~~
    \bs s_j^{(i)}=\bs H_j^{(i)}\bs E_j, ~~ i=1,\cdots,\mathcal{M}, ~~j=1,2 \label{4.3b}
\end{align}
\end{subequations}
are solutions to
%\begin{subequations}
\begin{align*}
& \bs K_1\bs M_1^{(i)}-\bs M_1^{(i)}\bs K_2=\bs r_1^{(i)}(\bs s_2^{(i)})^T, ~~
\bs K_2\bs M_2^{(i)}-\bs M_2^{(i)}\bs K_1=\bs r_2^{(i)}(\bs s_1^{(i)})^T,\\
& \bs r_{1,x_n}^{(i)}=a_1^{(i)}\bs K_1^n\bs r_1^{(i)}, ~~
\bs s_{1,x_n}^{(i)}=-a_2^{(i)}(\bs K_1^T)^n\bs s_1^{(i)}, \\
& \bs r_{2,x_n}^{(i)}=a_2^{(i)}\bs K_2^n\bs r_2^{(i)}, ~~
\bs s_{2,x_n}^{(i)}=-a_1^{(i)}(\bs K_2^T)^n\bs s_2^{(i)},
 ~~ i=1,\cdots,\mathcal{M}, ~~ n\in \mathbb{Z},
\end{align*}
%\end{subequations}
where
%\bsb\label{FFHH}
\begin{align*}
    &\bs F_1^{(i)}=\mr{diag}(\bs F_{n_1}[\rho^{(i)}(k_1)],\dots,\bs F_{n_p}[\rho^{(i)}(k_p)]),
     ~~\bs H_1^{(i)}=\mr{diag}(\bs H_{n_1}[\sigma^{(i)}(k_1)],\dots,\bs H_{n_p}[\sigma^{(i)}(k_p)]),\\
     &\bs F_2^{(i)}=\mr{diag}(\bs F_{m_1}[\rho^{(i)}(l_1)],\dots,\bs F_{m_q}[\rho^{(i)}(l_q)]),
     ~~\bs H_2^{(i)}=\mr{diag}(\bs H_{m_1}[\sigma^{(i)}(l_1)],\dots,\bs H_{m_q}[\sigma^{(i)}(l_q)]),
     %& \bs E_1=\bs E_{n_1n_2\cdots n_p},~~ \bs E_2=\bs E_{m_1m_2\cdots m_q},
\end{align*}
%\esb
$\bs E_1=\bs E_{n_1n_2\cdots n_p}, ~ \bs E_2=\bs E_{m_1m_2\cdots m_q}$,
and the plane wave factors $\rho^{(i)}(k)$ and $\sigma^{(i)}(k)$ are defined as in \eqref{rho-sij-APP}.
Consequently,
\begin{subequations}
\begin{align}
& \bs M_1=\sum_{i=1}^{\mathcal{M}}\bs M_1^{(i)}
= \sum_{i=1}^{\mathcal{M}} \bs F_1^{(i)}\bs G_1\bs H_2^{(i)}, ~~
 \bs M_2=\sum_{i=1}^{\mathcal{M}}\bs M_2^{(i)}
=- \sum_{i=1}^{\mathcal{M}}  \bs F_2^{(i)}\bs G_1^T\bs H_1^{(i)},\label{4.5a}\\
&\bs r_j=(\bs r_j^{(1)}, \bs r_j^{(2)}, \cdots, \bs r_j^{(\mathcal{M})})^T,~~
\bs s_j=(\bs s_j^{(1)}, \bs s_j^{(2)}, \cdots, \bs s_j^{(\mathcal{M})})^T,~~j=1,2 \label{4.5b}
\end{align}
\end{subequations}
are solutions to \eqref{symm-case} and \eqref{DR-symm}.
Note that  $\bs K_1, \bs C_1, \bs F^{(i)}_1 \in \mathcal{T}_{n_1n_2\cdots n_p}$,
$\bs H^{(i)}_1 \in \mathcal{H}_{n_1n_2\cdots n_p}$,
$\bs K_2, \bs C_2, \bs F^{(i)}_2 \in \mathcal{T}_{m_1m_2\cdots m_q}$
and $\bs H^{(i)}_2 \in \mathcal{H}_{m_1m_2\cdots m_q}$.
%For more details of the involved  notations one can refer to Appendix \ref{APP-A}.
\end{theorem}

In the symmetric Sylvester equation case, $\bs S^{(i,j)}$  is defined as\footnote{
Note that when $\bs C$ is invertible, \eqref{Sij-Sy} can be deformed into
$\bs S^{(i,j)}=\bs s^T\bs K^j (\bs I+\bs M')^{-1}\bs K^i\bs r'$
by redefining $\bs M=\bs C \bs M'$ and $\bs r=\bs C \bs r^{\prime}$.
We will use \eqref{Sij-Sy} for the sack of generality.}
\begin{align}\label{Sij-Sy}
        \bs S^{(i,j)} \doteq \bs s^T\bs K^j (\bs C+\bs M)^{-1}\bs K^i\bs r, ~~ i,j\in \mathbb{Z}.
\end{align}
In   $2\times 2$ block matrix form
\begin{align}\label{Sij-2}
    \bs S^{(i,j)}=
    \begin{pmatrix}
        \bs s_1^{(i,j)} & \bs s_2^{(i,j)} \\
        \bs s_3^{(i,j)} & \bs s_4^{(i,j)}
    \end{pmatrix},
\end{align}
where each $\bs s_l^{(i,j)}$ is a $\mathcal{M}\times \mathcal{M}$ matrix,
we have formulae
\bsb\label{sij}
\begin{align}
    \bs s_1^{(i,j)}&=-\bs s_2^T\bs K_2^j\bs C_2^{-1}\bs M_2(\bs C_1-\bs M_1\bs C_2^{-1}\bs M_2)^{-1}
    \bs K_1^i\bs r_1, \\
    \bs s_2^{(i,j)}&=\bs s_2^T\bs K_2^j(\bs C_2-\bs M_2\bs C_1^{-1}\bs M_1)^{-1}\bs K_2^i\bs r_2, \label{s2}\\
    \bs s_3^{(i,j)}&=\bs s_1^T\bs K_1^j(\bs C_1-\bs M_1\bs C_2^{-1}\bs M_2)^{-1}\bs K_1^i\bs r_1, \\
    \bs s_4^{(i,j)}&=-\bs s_1^T\bs K_1^j\bs C_1^{-1}\bs M_1(\bs C_2-\bs M_2\bs C_1^{-1}\bs M_1)^{-1}
    \bs K_2^j\bs r_2.\label{sij-4}
\end{align}
\esb
At this moment we need to establish the relations between $\bs S^{(i,j)}$ and $\bs S^{(j,i)}$,
which will be used to discuss Hermitian property of $\bs V$.

\begin{lemma}\label{L-4-1}
In the case for symmetric Sylvester equation, we have (symmetric) relations:
\begin{align}\label{ssij-sym}
        (\bs s_1^{(i,j)})^T=-\bs s_4^{(j,i)}, ~~ (\bs s_2^{(i,j)})^T=\bs s_2^{(j,i)},
        ~~ (\bs s_3^{(i,j)})^T=\bs s_3^{(j,i)}.
\end{align}
Note that when $\mathcal{M}=1$   the relation \eqref{ssij-sym} reduces to the scalar case
that we found in \cite{LQYZ-SAPM-2022}.
\end{lemma}

\begin{proof}
As an example we prove $(\bs s_2^{(i,j)})^T=\bs s_2^{(j,i)}$. The other two relations in \eqref{ssij-sym}
can be proved similarly.
$\bs s_2^{(i,j)}$ is a $\mathcal{M}\times \mathcal{M}$  matrix defined by \eqref{s2}.
Its $(\alpha,\beta)$-th element, denoted as $(\bs s_2^{(i,j)})_{\alpha,\beta}$, is
\begin{equation}\label{s2-ab-1}
(\bs s_2^{(i,j)})_{\alpha,\beta}=(\bs s_2^{(\alpha)})^T\bs K_2^j
    (\bs C_2-\bs M_2\bs C_1^{-1}\bs M_1)^{-1}\bs K_2^i\bs r_2^{(\beta)},
\end{equation}
where $\bs s_2^{(\alpha)}$ and $\bs r_2^{(\beta)}$ are column vectors of $\bs s_2$ and $\bs r_2$
(see \eqref{4.3b} and \eqref{4.5b}),
and $\bs M_1, \bs M_2$ are given in \eqref{4.5a}.
Substituting \eqref{4.3b} and \eqref{4.5a} into \eqref{s2-ab-1} yields
\begin{equation*}
(\bs s_2^{(i,j)})_{\alpha,\beta}=
\bs E_2^T\bs H_2^{(\alpha)}\bs K_2^j
\left(\bs C_2+\sum_{s,l=1}^{\mathcal{M}}\bs F_2^{(s)}\bs G_1^T\bs H_1^{(s)}\bs C_1^{-1}
    \bs F_1^{(l)}\bs G_1\bs H_2^{(l)}\right)^{-1}\bs K_2^i\bs F_2^{(\beta)}\bs E_2.
\end{equation*}
Noticing that $\bs K_2, \bs C_2, \bs F^{(i)}_2 \in \mathcal{T}_{m_1m_2\cdots m_q}$
and $\bs H^{(i)}_2 \in \mathcal{H}_{m_1m_2\cdots m_q}$, in light of Remark \ref{R-A-1} (in Appendix \ref{APP-A}),
we have
\[\bs H^{(i)}_2\bs K_2=\bs K_2^T\bs H^{(i)}_2, ~~
\bs K_2 \bs F^{(i)}_2=\bs F^{(i)}_2\bs K_2. \]
It follows that
\begin{equation*}
(\bs s_2^{(i,j)})_{\alpha,\beta}=
\bs E_2^T(\bs K_2^j)^T \bs A^{-1}_{(\alpha,\beta)}\bs K_2^i \bs E_2,
\end{equation*}
where
\begin{equation*}
\bs A_{(\alpha,\beta)}=(\bs H_2^{(\alpha)}\bs C_2^{-1}\bs F_2^{(\beta)})^{-1}
+\sum_{s,l=1}^{\mathcal{M}}(\bs F_2^{(\beta)})^{-1}\bs F_2^{(s)}\bs G_1^T\bs H_1^{(s)}\bs C_1^{-1}
    \bs F_1^{(l)}\bs G_1\bs H_2^{(l)}(\bs H_2^{(\alpha)})^{-1}.
\end{equation*}
Again, since  $\bs K_1, \bs C_1, \bs F^{(i)}_1 \in \mathcal{T}_{n_1n_2\cdots n_p}$,
$\bs H^{(i)}_1 \in \mathcal{H}_{n_1n_2\cdots n_p}$,
$\bs K_2, \bs C_2, \bs F^{(i)}_2 \in \mathcal{T}_{m_1m_2\cdots m_q}$
and $\bs H^{(i)}_2 \in \mathcal{H}_{m_1m_2\cdots m_q}$,
according to Remark \ref{R-A-1},
both $\bs H_2^{(\alpha)}\bs C_2^{-1}\bs F_2^{(\beta)}$ and
$\bs H_1^{(s)}\bs C_1^{-1} \bs F_1^{(l)}$ are symmetric mareices.
In addition, from \eqref{HF} we have $\bs H_2^{(\alpha)}\bs F_2^{\beta)}=\bs H_2^{(\beta)}\bs F_2^{(\alpha)}$,
which yields
\[(\bs H_2^{(\beta)})^{-1}\bs H_2^{(s)}
=\bs F_2^{(s)}(\bs F_2^{(\beta)})^{-1}
=(\bs F_2^{(\beta)})^{-1}\bs F_2^{(s)}.\]
Thus, $\bs A$ is written as
\begin{equation*}%\label{A}
\bs A_{(\alpha,\beta)}=\left(\bs H_2^{(\alpha)}\bs C_2^{-1}\bs F_2^{(\beta)}\right)^{-1}
+\sum_{s,l=1}^{\mathcal{M}}(\bs H_2^{(\beta)})^{-1}\bs H_2^{(s)}\bs G_1^T
\left(\bs H_1^{(s)}\bs C_1^{-1}\bs F_1^{(l)}\right)\bs G_1\bs H_2^{(l)}(\bs H_2^{(\alpha)})^{-1}.
\end{equation*}
%which is a symmetric matrix.
In a same way, we have
\begin{equation*}
(\bs s_2^{(j,i)})_{\beta,\alpha}=
\bs E_2^T(\bs K_2^i)^T \bs A^{-1}_{(\beta,\alpha)}\bs K_2^j \bs E_2.
\end{equation*}
Next, using relation \eqref{HF} once again, we have
$\bs H_2^{(\alpha)}\bs F_2^{\beta)}=\bs H_2^{(\beta)}\bs F_2^{(\alpha)}$ and then
\[\bs H_2^{(\alpha)}\bs C_2^{-1}\bs F_2^{(\beta)}=
\bs H_2^{(\alpha)}\bs F_2^{(\beta)}\bs C_2^{-1}
=\bs H_2^{(\beta)}\bs F_2^{(\alpha)}\bs C_2^{-1}
=\bs H_2^{(\beta)}\bs C_2^{-1}\bs F_2^{(\alpha)},\]
where we have also made use of the fact that
$\bs C_2, \bs F^{(\alpha)}_2, \bs F^{(\beta)}_2$ belong to the Abelian group $\mathcal{T}_{m_1m_2\cdots m_q}$.
Similarly, we have $\bs H_1^{(s)}\bs C_1^{-1}\bs F_1^{(l)}=\bs H_1^{(l)}\bs C_2^{-1}\bs F_2^{(s)}$.
This means
\begin{equation*}
\bs A_{(\alpha,\beta)}=\bs A^T_{(\beta,\alpha)}.
\end{equation*}
Finally, note that $(\bs s_2^{(i,j)})_{\alpha,\beta}$ is a scalar function, we immediately arrive at
\[(\bs s_2^{(i,j)})_{\alpha,\beta}=((\bs s_2^{(i,j)})_{\alpha,\beta})^T
=(\bs s_2^{(j,i)})_{\beta,\alpha},\]
which gives rise to the relation $(\bs s_2^{(i,j)})^T=\bs s_2^{(j,i)}$.
The other two relations in \eqref{ssij-sym} can be proved in a similar way.
We skip the details.

\end{proof}

\subsection{Reduction to \eqref{SDYM}}\label{sec-4-2}

Consider the constraints
\begin{align}\label{conjugate_con}
    \bs K_2=-\bar{\bs K}_1^{-1}, ~~\bs C_2=\bar{\bs C}_1, ~~\bs a_2=-\bar{\bs a}_1,
\end{align}
under which it turns out that \eqref{DR-symm-2} allows solutions
\begin{equation}\label{rs2}
\bs r_2=\bar{\bs K}_1^{-1}\bar{\bs r}_1, ~~
\bs s_2=-(\bs K_1^\dagger)^{-1}\bar{\bs s}_1,
\end{equation}
provided that $\bs r_1$ and $\bs s_1$ satisfy \eqref{DR-symm-1}.
Then, due to the uniqueness of solutions to the Sylvester equations \eqref{symm-case},
by a similar manner to the asymmetric case (see Sec.\ref{sec-3-2}), we can find
\begin{equation}\label{MM}
\bs M_2=\bar{\bs M}_1.
\end{equation}
Thus, the two Sylvester equations in \eqref{symm-case} reduce to one equation
\begin{equation}\label{Syl-sym}
\bs K_1 \bs M_1 \bar {\bs K}_1 +\bs M_1 =- \bs r_1  \bs s^\dag_1.
\end{equation}
In addition, similar to \eqref{s} and \eqref{r} in the asymmetric case, we have
\begin{equation*}%\label{rs}
\bs r_{\bar{x}_n}=(-1)^{n+1}\bs r_{x_{-n}},~~
\bs s_{\bar{x}_n}=(-1)^{n+1}\bs s_{x_{-n}}.
\end{equation*}
This indicates that the coordinates constraint \eqref{x} is available as well to this case,
which leads to the coordinate formulation \eqref{coor} and further
leads to equation \eqref{SDYM-4} which is in the same form of the SDYM equation \eqref{SDYM}.

It has been proved in Theorem \ref{T-2-2} that $|\bs V|=1$ in the symmetric case.
Next, we look at Hermitian property of $\bs V$ under the constraint \eqref{conjugate_con}.
For convenience, we write $\bs V$ in the form
\begin{align}\label{V-4.10}
\bs V=
\begin{pmatrix}
  \bs v_1 & \bs v_2 \\
  \bs v_3 & \bs v_4
\end{pmatrix}.
\end{align}
It follows from the expressions $\bs V=\bs I_{2\mathcal{M}}-\bs S^{(-1,0)}$ and \eqref{sij},
and the relations \eqref{conjugate_con}, \eqref{rs2} and \eqref{MM}  that
\begin{align*}
    \bs v_1=\bs I_{\mathcal{M}}-\bs s_1^{(-1,0)}
    &=\bs I_{\mathcal{M}}
    +\bs s_2^T\bs C_2^{-1}\bs M_2(\bs C_1-\bs M_1\bs C_2^{-1}\bs M_2)^{-1}\bs K_1^{-1}\bs r_1 \\
    &=\bs I_{\mathcal{M}}-\bs s_1^\dagger(\bar{\bs K}_1)^{-1}\bar{\bs C}_1^{-1}\bar{\bs M}_1
    (\bs C_1-\bs M_1\bar{\bs C}_1^{-1}\bar{\bs M}_1)^{-1}\bs K_1^{-1}\bs r_1.
\end{align*}
Then we find
\begin{align*}
\bar{\bs v}_1 =\bs I_{\mathcal{M}}-\bs s_1^T(\bs K_1)^{-1}\bs C_1^{-1}
\bs M_1(\bs C_2-\bs M_2\bs C_1^{-1}
\bs M_1)^{-1}\bs r_2 =\bs I_{\mathcal{M}}+\bs s_4^{(0,-1)},
\end{align*}
where we have used \eqref{sij-4}.
Recalling relation \eqref{ssij-sym}, we have
\begin{equation*}
\bar{\bs v}_1=\bs I_{\mathcal{M}}-(\bs s_1^{(-1,0)})^T=\bs v_1^T.
\end{equation*}
In a similar way we can find that $\bar{\bs v}_4=\bs v_4^T$.
As for $\bs v_2$, we have
\begin{align*}
    \bs v_2&=-\bs s_2^T(\bs C_2-\bs M_2\bs C_1^{-1}\bs M_1)^{-1}\bs K_2^{-1}\bs r_2
    =-\bs s_1^{\dag}\bar{\bs K}_1^{-1}
    (\bar{\bs C}_1-\bar{\bs M}_1\bs C_1^{-1}\bs M_1)^{-1}\bar{\bs r}_1
\end{align*}
and
\begin{align*}
    \bar{\bs v}_2 &=-\bs s_1^T \bs K_1^{-1}(\bs C_1-\bs M_1
    \bs C_2^{-1}\bs M_2)^{-1}\bs r_1
     =-\bs s_3^{(0,-1)}=-(\bs s_3^{(-1,0)})^T=\bs v_3^T.
\end{align*}
This means $\bs V$ is a Hermitian matrix.

In conclusion, we have the following.

\begin{theorem}\label{T-4-2}
In the symmetric Sylvester equation case, we have $|\bs V|=1$.
In addition, with the constraint \eqref{conjugate_con} and coordinates $\{y_n\}$ defined in \eqref{coor},
the matrix function $\bs V$ in the form \eqref{V-4.10} with
\begin{subequations}\label{V-vv}
\begin{align}
\bs v_1
    &=\bs I_{\mathcal{M}}-\bs s_1^\dagger(\bar{\bs K}_1)^{-1}\bar{\bs C}_1^{-1}\bar{\bs M}_1
    (\bs C_1-\bs M_1\bar{\bs C}_1^{-1}\bar{\bs M}_1)^{-1}\bs K_1^{-1}\bs r_1,\label{v1-sym}\\
\bs v_2
    &=-\bs s_1^{\dag}\bar{\bs K}_1^{-1}
    (\bar{\bs C}_1-\bar{\bs M}_1\bs C_1^{-1}\bs M_1)^{-1}\bar{\bs r}_1,
    ~~~ \bs v_3= \bs v_2^\dag,\\
\bs v_4
    &=\bs I_{\mathcal{M}}-\bs s_1^T \bs C_1^{-1}\bs M_1
    (\bar{\bs C}_1-\bar{\bs M}_1\bs C_1^{-1}\bs M_1)^{-1}
    \bar{\bs r}_1,
\end{align}
\end{subequations}
satisfies equation \eqref{SDYM-4}
and $J=\bs V$ is a solution to the SDYM equation \eqref{SDYM} with
$y=y_n, z=\bar{y}_{n+1}$.
$\bs V$ is a Hermitian matrix.
\end{theorem}

\subsection{Positive definiteness}\label{sec-4-3}

Theorem \ref{T-4-2} shows that $\bs V$ is a Hermitian matrix and $|\bs V|=1$.
In the following we discuss  the positive definiteness of $\bs V$ for the
SU(2) SDYM equation.
We consider its one-soliton and two-soliton solutions.
These solutions are the same as those obtained in \cite{LQYZ-SAPM-2022}.
Note that  for the SU(2) SDYM equation which corresponds to $\mathcal{M}=1$,
in the following discussion we will drop off the superscript ${}^{(1)}$
from $\rho_j^{(1)}, \sigma_j^{(1)}, \overset{\circ}{\rho}_j\!{}^{(1)}, \overset{\circ}{\sigma}_j\!{}^{(1)}$
and $a^{(1)}$ without making confusion.

\subsubsection{One-soliton solution of the SU(2) SDYM equation}\label{sec-4-3-1}

For one-soliton solution of the SU(2) SDYM equation, we have
$\mathcal{M}=1$ and $M=1$.
In this case, $\bs V$ is a $2\times 2$ matrix defined as \eqref{V-4.10}
where each $\bs v_j$ is a scalar function.
Since $|\bs V|=1$, the positive definiteness of $\bs V$ is determined by the sign of $\bs v_1$,
which is given by \eqref{v1-sym}.

For one-soliton solution, we have
\begin{align}\label{krsm}
    \bs K_1=k_1, ~~\bs r_1=\rho_1, ~~\bs s_1=\sigma_1, ~~\bs a_1=a, ~~\bs C_1=c_1, ~~
    \bs M_1=m_1=-\frac{\rho_1\bar\sigma_1}{|k_1|^2+1},
\end{align}
where (for $j=1$)
%\bsb
\begin{align}\label{rho-sig}
    &\rho_j=\overset{\circ}{\rho}_j\,\exp[a\mathfrak L_j(\bs x)],~~~
    \sigma_j=\overset{\circ}{\sigma}_j\,\exp[{\bar a}\mathfrak L_j(\bs x)],
\end{align}
%\esb
where $\mathfrak L_j(\bs x)$ is defined as in \eqref{L(x)}.
It then follows from \eqref{v1-sym} that
\begin{align}\label{v1-1SS}
    \bs v_{1}=\frac{|c_1|^2+|m_1|^2|k_1|^{-2}}{|c_1|^2-|m_1|^2}
    =-\frac{1}{|k_1|^2}+\big(1+\frac{1}{|k_1|^2}\big)\frac{|c_1|^2}{|c_1|^2-|m_1|^2},
\end{align}
where (we take $a=1$ for convenience)
\begin{align*}
    |m_1|
    =\frac{|\overset{\circ}{\rho}_1\overset{\circ}{\sigma}_1|}{|k_1|^2+1}
    \exp\Big(2\mr{Re}[\mathfrak L_1(\bs x)]\Big),
\end{align*}
and $\mathfrak L_1(\bs x)$ is defined as in \eqref{L(x)}.
\begin{figure}[ht!]
  \begin{center}
\begin{tikzpicture}[scale=.65]
        %\draw[step=.5cm, gray, thin, dotted] (-3,-3) grid (3,3);
        \draw[->] (-4,0) -- (4.5,0);
        \draw (3.9,-0.5) node {$|m_1|$};
        \draw[->] (-3.0,-3.5) -- (-3.0,3.5);
        \draw (-3.5,2.8) node {$\bs v_1$};
        \draw[thick,dashed] (-1.4,-3.4) -- (-1.4,3.4);
        \draw (-0.8,0.4) node {$|c_1|$};
        \draw[thick,dashed] (-3.4,-1.4) -- (3.4,-1.4);
        \draw (-4.2,-1.4) node {$-\frac{1}{|k_1|^2}$};
        \draw (-3.3,1.0) node {$1$};
        \draw[thick] (-3,1) .. controls (-1.8,1.2) and (-1.7,2) ..(-1.5,3.4);
        \draw[thick] (-1.3,-3.4) .. controls (-1,-2) and (1,-1.4) .. (3.4,-1.5);
\end{tikzpicture}
\end{center}
\caption{Shape of $\bs v_1$ defined in \eqref{v1-1SS}}
\label{Fig-1}
\end{figure}
Fig.\ref{Fig-1} shows that how $\bs v_1$ varies with  $|m_1|$:
$\bs v_1$ is positive when $|m_1|$ is less than $|c_1|$,
while $\bs v_1$ is negative when $|m_1|$ is greater than $|c_1|$.
Equation $|m_1|=|c_1|$ defines a hyper plane
\begin{align*}
    \mr{Re}[\mathfrak L_1(\bs x)]
    =\frac{1}{2}\ln\Big(\frac{|c_1|(|k_1|^2+1)}{|\overset{\circ}{\rho}_1\overset{\circ}{\sigma}_1|}\Big),
\end{align*}
which divides $\mathbb{R}^4$ into two parts:
one is the domain $D_1$ where
$\mr{Re}[\mathfrak L_1(\bs x)]
    <\frac{1}{2}\ln\Big(\frac{|c_1|(|k_1|^2+1)}{|\overset{\circ}{\rho}_1\overset{\circ}{\sigma}_1|}\Big)$
and $\bs v_1$ is positive, the other is the domain $D_2$ where
$\mr{Re}[\mathfrak L_1(\bs x)]
    >\frac{1}{2}\ln\Big(\frac{|c_1|(|k_1|^2+1)}{|\overset{\circ}{\rho}_1\overset{\circ}{\sigma}_1|}\Big)$
and $\bs v_1$ is negative.
Thus, $\bs V$ is positive-definite when $\bs x \in D_1$ and   negative-definite when $\bs x \in D_2$.

\begin{remark}\label{R-4-1}
For the SU(2) SDYM equation, when $\bs x \in D_2$,  $\bs V$ and $J=-\bs V$
are respectively  negative-definite and positive-definite.
Then one can take a piecewise definition for $J$:
\[J=\left\{\begin{array}{ll} \bs V, & \bs x \in D_1,\\
                                         -\bs V, & \bs x \in D_2.
       \end{array}\right.
\]
\end{remark}
\begin{remark}\label{R-4-1}
The one-soliton solution in this section is different from the one discussed in Sec.\ref{sec-3-3-1},
cf. Proposition \ref{P-5-1}.
\end{remark}

\subsubsection{Two-soliton  solution of the  SU(2) SDYM}\label{sec-4-3-2}

For the SU(2) SDYM equation, for arbitrary $M$,
$\bs v_1$ is always a scalar function and can be expressed as
\begin{align}\label{v1}
\bs v_{1}=\frac{g}{f}
    =\frac{|\bs C_1\bar {\bs C}_1+\bs K_1^{-1} \bs M_1 \bar {\bs K}_1^{-1}\bar {\bs C}_1^{-1}
    \bar {\bs M}_1\bar {\bs C}_1|}
    {|\bs C_1\bar {\bs C}_1-\bs M_1 \bar {\bs C}_1^{-1}\bar {\bs M}_1\bar {\bs C}_1|}.
\end{align}
In fact,  for arbitrary $M$, from equation \eqref{v1-sym} we have
\begin{align*}
\bs v_1
    &=1-\bs s_1^\dagger \bar{\bs K}_1^{-1}\bar{\bs C}_1^{-1}\bar{\bs M}_1
    (\bs C_1-\bs M_1\bar{\bs C}_1^{-1}\bar{\bs M}_1)^{-1}\bs K_1^{-1}\bs r_1,\\
    &=|\bs I_M-    (\bs C_1-\bs M_1\bar{\bs C}_1^{-1}\bar{\bs M}_1)^{-1}\bs K_1^{-1}\bs r_1
    \bs s_1^\dagger \bar{\bs K}_1^{-1}\bar{\bs C}_1^{-1}\bar{\bs M}_1|,
\end{align*}
where we have made use of the Weinstein--Aronszajn formula (see \eqref{WA}).
Next, replacing $\bs r_1 \bs s_1^\dagger$ using \eqref{Syl-sym} yields
\begin{align*}
\bs v_1
    &=|\bs I_M+    (\bs C_1-\bs M_1\bar{\bs C}_1^{-1}\bar{\bs M}_1)^{-1}\bs K_1^{-1}
    (\bs K_1 \bs M_1 \bar{\bs K}_1+\bs M_1)\bar{\bs K}_1^{-1}\bar{\bs C}_1^{-1}\bar{\bs M}_1| \\
    &=|(\bs C_1-\bs M_1\bar{\bs C}_1^{-1}\bar{\bs M}_1)^{-1}
     (\bs C_1+\bs K_1^{-1} \bs M_1 \bar {\bs K}_1^{-1}\bar {\bs C}_1^{-1}\bar {\bs M}_1)| \\
    &= \frac{|\bs C_1+\bs K_1^{-1} \bs M_1 \bar {\bs K}_1^{-1}\bar {\bs C}_1^{-1}\bar {\bs M}_1|}
    {|\bs C_1-\bs M_1\bar{\bs C}_1^{-1}\bar{\bs M}_1|}.
\end{align*}
After multiplying $|\bar {\bs C}_1|$ on both numerator and denominator, we get \eqref{v1}.

For the case of two-soliton solution, $M=2$ and  we have
\begin{align}
  &  \bs K_1=\begin{pmatrix}
        k_1 & 0\\
        0 & k_2
    \end{pmatrix},~~~
    \bs C_1=\begin{pmatrix}
        c_1 & 0 \\
        0 & c_2
    \end{pmatrix},~~~
    \bs r_1=\begin{pmatrix}
        \rho_{1}  \\
        \rho_{2}
    \end{pmatrix},~~~~
    \bs s_1=\begin{pmatrix}
        \sigma_{1}  \\
        \sigma_{2}
    \end{pmatrix},\\
 &   \bs M_1
 %=\begin{pmatrix}
%         m_{11} &  m_{12} \\
%         m_{21} &  m_{22}
%    \end{pmatrix}
    =\begin{pmatrix}
        -\rho_{1}\bar\sigma_{1}/(|k_1|^2+1) &- \rho_{1}\bar\sigma_{2}/(k_1\bar k_2+1) \\
        -\rho_{2}\bar\sigma_{1}/(\bar k_1k_2+1) & -\rho_{2}\bar\sigma_{2}/(|k_2|^2+1)
    \end{pmatrix},
\end{align}
where $\rho_j$ and $\sigma_j$ are defined in \eqref{rho-sig} for $j=1,2$,
and we already dropped off the superscript ${}^{(1)}$
from $\rho_j^{(1)}, \sigma_j^{(1)}$
and $a^{(1)}$ without making confusion.
In the following, for convenience we take $c_1=c_2=1$ and rewrite $\bs v_1$ in terms of
\begin{equation}\label{w-j2}
w_j = \rho_j\sigma_j
=\overset{\circ}{w}_j \exp\bigl[(a+\bar a)\mathfrak L_j(\bs x)\bigr], ~~
\overset{\circ}{w}_j= \overset{\circ}{\rho}_j\overset{\circ}{\sigma}_j,~~ j=1,2.
\end{equation}
It follows that
\begin{align}\label{sym-2ss-v1}
   \bs v_{1}=\frac{1+\frac{\bigl|\frac{w_1}{k_1}\bigr|^2}{(1+|k_1|^2)^2}
   +\frac{\bigl|\frac{w_2}{k_2}\bigr|^2}{(|k_2|^2+1)^2}
   +\frac{|k_1-k_2|^4\bigl|\frac{w_1w_2}{k_1k_2}\bigr|^2}
   {(|k_1|^2+1)^2(|k_2|^2+1|)^2 |1+k_1\bar k_2|^4}
   +2\mr{Re}\bigl[\frac{ w_1\bar w_2}{k_1\bar k_2 (1+ k_1\bar k_2)^2}\bigr]}
   {1-\frac{|w_1|^2}{(|k_1|^2+1)^2}-\frac{|w_2|^2}{(|k_2|^2+1)^2}
 +\frac{|k_1-k_2|^4|w_1w_2|^2}{(|k_1|^2+1)^2(|k_2|^2+1|)^2 |1+k_1\bar k_2|^4}
 -2\mathrm{Re}\bigl[\frac{w_1 \bar w_2}{(1+k_1\bar k_2)^2}\bigr]}.
\end{align}
Note that $w_j$ defined in \eqref{w-j2} and \eqref{w-x}
essentially have the same expression in light of \eqref{3.19}
and the real coefficients $(a^{(1)}-a^{(2)})$ and $(a+\bar a)$
can always be normalized to be $1$.
Again, we may redefine $w_j$ as
\begin{equation}\label{wj-til-2}
w_j=(1+|k_j|^2) \widetilde{w}_j, ~~ j=1,2,
\end{equation}
and rewrite the above $\bs v_1$ in terms of $\widetilde{w}_j$  as
\begin{equation}\label{v1-wt}
\bs v_{1}=\frac{1+\bigl|\frac{\widetilde{w}_1}{k_1}\bigr|^2+\bigl|\frac{\widetilde{w}_2}{k_2}\bigr|^2
+ \frac{|k_1-k_2|^4}{|1+k_1\bar k_2|^4}\bigl|\frac{\widetilde{w}_1\widetilde{w}_2}{k_1k_2}\bigr|^2
+2 (1+|k_1|^2)(1+|k_2|^2 )
\mathrm{Re}\Bigl[\frac{ \widetilde{w}_1\bar{\widetilde{w}}_2}{k_1\bar k_2(1+k_1\bar k_2)^2}\Bigr] }
{1-| \widetilde{w}_1|^2-| \widetilde{w}_2|^2
+ \frac{|k_1-k_2|^4}{|1+k_1\bar k_2|^4}|\widetilde{w}_1\widetilde{w}_2|^2
-2 (1+|k_1|^2)(1+|k_2|^2)
\mathrm{Re}\Bigl[\frac{\widetilde{w}_1 \bar{\widetilde{w}}_2}{(1+k_1\bar k_2)^2}\Bigr] }.
\end{equation}
This is nothing but $v_{11}$ given in \eqref{v11-wt},
and thus, the definite-positive domain of  $\bs V$
can be determined in the same way as in Sec.\ref{sec-3-3-2}.

To compare with the two-soliton solution obtained from the asymmetric Sylveser formulation (see Sec.\ref{sec-3-3-2}),
we also need to calculate $\bs v_2$ and $\bs v_4$ from \eqref{V-vv} (with $\bs C_1=\bs I_2$)
and write out them in terms of $\widetilde w_j$ defined in \eqref{wj-til-2}.
It turns out  that 
\begin{equation}\label{v-v}
\left(\begin{array}{cc}
      \bs v_1 & \bs v_2\\
      \bar{\bs v}_2 & \bs v_4
      \end{array}\right)
=\left(\begin{array}{cc}
      1 & 0\\
      0 & \bar k_1 \bar k_2
      \end{array}\right)
      \left(\begin{array}{cc}
      v_{11} & v_{12}\\
      \bar{v}_{12} & v_{22}
      \end{array}\right)
      \left(\begin{array}{cc}
      1 & 0\\
      0 & k_1 k_2
      \end{array}\right),
\end{equation}
where $v_{11}$, $v_{12}$ and $v_{22}$ are given as in \eqref{v11-wt}, \eqref{v12-wt} and \eqref{v22-wt}.
Note that for any solution $J$ to the SDYM equation \eqref{SDYM},
$J'=PJQ$ is still a solution provided $P$ and $Q$ are invertible constant matrices.
In this sense, we say the two-solition solutions obtained in the asymmetric and symmetric
Sylvester formulations are same.

\section{Concluding remarks}\label{sec-5}

In this paper we have established two Cauchy matrix schemes for the SU($\mathcal{N}$) SDYM equation \eqref{SDYM}.
These schemes are based on two different Sylvester equations (asymmetric and symmetric cases),
both of which can generate the noncommutative relations \eqref{Sij_moves} and \eqref{diff_sum_pro}
that are used to obtain the unreduced equation \eqref{SDYM-3}.
This equation was then reduced to the  SU($\mathcal{N}$) SDYM equation
(in Yang's formulation \cite{Yang-1977,BFNY-1978}) in Sec.\ref{sec-3} and Sec.\ref{sec-4},
for the asymmetric and symmetric Sylvester equation cases, respectively.
As solutions, in both cases $\bs V$ is Hermitian and $|\bs V|$ is a constant.
The property of positive-definiteness of $\bs V$
was investigated.
For the SU(2) SDYM equation and the obtained one-soliton solutions and two-soliton solutions,
we have worked out a way to obtain a domain in $\mathbb{R}^4$ where $\bs V$ is  positive-definite.
Note that our solutions are different from those of the SU($\mathcal{N}$) SDYM equation
obtained using Darboux transformation method \cite{NGO-2000}.
In the two Cauchy matrix schemes, the dispersion relation of solutions can be characterized
by $\mathfrak L_j(\bs x)$ defined in \eqref{Lcx},
which contains an arbitrary ``$n$''. Such an arbitrariness allows us to have dimension reductions
with respect of $\bs x$, for example, to get solutions to a 3-dimensional
relativistic-invariant system studied by Manakov and Zakharov \cite{MZ-1981}.
We will present the example with more details in Appendix \ref{APP-B}.

To derive  equation \eqref{SDYM-3},
we also provided an alternative approach in Appendix \ref{APP-A},
where we made use of matrix equation $\bs\Phi_{x_n}=\bs \Phi\bs A\bs P_n$
and the compatibility with $\bs\Phi_{x_{n+1}}=\bs \Phi\bs A\bs P_{n+1}$.
Since $n$ is arbitrary, one can also consider the compatibility between
\begin{align*}
        \bs\Phi_{x_n}=\bs \Phi\bs A\bs P_n, ~~ \bs\Phi_{x_{n+l}}=\bs \Phi\bs A\bs P_{n+l}, ~~~l=2,\cdots,
\end{align*}
and derive new equations. For the case $l$ greater than one, more
$\{\bs S^{(i,j)}\}$ than $\bs S^{(0,0)}$ and $\bs S^{(-1,0)}$ will be needed.
In addition, it is worthy to mention that the two Cauchy matrix schemes that are involved
are respectively based on the Sylvester equation for the (matrix)
KP system (cf.\cite{FZ-KP-2022})
and  the Sylvester equation for the (matrix) AKNS system (cf.\cite{Zhao-2018}).
Since $\bs V$ is defined via $\bs S^{(-1,0)}$, this implies possible links
between the SDYM equation and the equations related to $\bs S^{(-1,0)}$ element
in the matrix KP hierarchy and matrix AKNS hierarchy.
In addition, in our paper, we reduce equation \eqref{SDYM-3} to the SU($\mathcal{N}$) SDYM equation in
$\mathbb{R}^4$.
It is also possible to consider other reductions in different metric spaces.
All these will be the topics for the future investigation.

\vskip 20pt
\subsection*{Acknowledgements}

The authors are grateful to the referee for the invaluable comments.
This work is supported by the National Natural Science Foundation of China
(nos. 12271334, 12126352, 11971251, 11875040, 11631007) and Science and Technology
Innovation Plan of Shanghai (20590742900).

\begin{appendices}

\section{An alternative way to equation \eqref{diff_recur}}\label{APP-A}

We have shown in Theorem \ref{T-2-1} (in Sec.\ref{sec-2-3})
that the two cases of the Sylvester equations
can give rise to the relations \eqref{Sij_moves} and \eqref{diff_sum_pro} for $\{\bs S^{(i,j)}\}$,
from which and along the lines of the treatment in \cite{LQYZ-SAPM-2022},
one can obtain equation \eqref{diff_recur} and then \eqref{SDYM-3}.

In the following we  present an alternative way to obtain equation \eqref{diff_recur}.
This will be valid for the two cases of the Sylvester equations discussed in Theorem \ref{T-2-1}.
Note again that in these two cases the relations \eqref{Sij_moves} and \eqref{diff_sum_pro} hold.
Introduce an auxiliary vector functions $\bs\phi^{(i)}\in\mb C_{N\times\mathcal{N}}[\mf x]$  defined by
\begin{align}\label{phi-i}
    \bs\phi^{(i)}=(\bs C+\bs M)^{-1}\bs K^i\bs r, ~~ i\in \mathbb{Z}.
\end{align}
It can be proved that
\begin{align*}
    \bs\phi^{(i)}_{x_n}+(\bs C+\bs M)^{-1}\bs M_{x_n}\bs\phi^{(i)}=\bs\phi^{(i+n)}\bs a.
\end{align*}
Substituting  \eqref{M_evo} into it and using the definition \eqref{Sij},
we can obtain
%\bsb\label{phi_evo}
\begin{align*}
    \bs\phi^{(i)}_{x_n}&=\bs\phi^{(i+n)}\bs a-\sum_{l=0}^{n-1}\bs\phi^{(n-1-l)}\bs a\bs S^{(i,l)},
    ~~~ (n\in\mb Z^+), \\
    \bs\phi^{(i)}_{x_0}&=\bs\phi^{(i)}\bs a,      \\
    \bs\phi^{(i)}_{x_n}&=\bs\phi^{(i+n)}\bs a+\sum_{l=-1}^{n}\bs\phi^{(n-1-l)}\bs a\bs S^{(i,l)},
    ~~~ (n\in\mb Z^-).
\end{align*}
%\esb
These relations can be expressed in a matrix form
\begin{align}\label{Phi-n}
    \bs\Phi_{x_n}=\bs\Phi\bs A \bs P_n,~~ n\in \mathbb{Z},
\end{align}
where $\bs \Phi$ is a $N\times \infty$ matrix composed by column vectors \eqref{phi-i},
\[\bs\Phi=(\cdots,\bs\phi^{(-1)},\bs\phi^{(0)},\bs\phi^{(1)},\cdots),\]
$\bs A=\mr{diag}(\cdots,\bs a,\bs a,\bs a,\cdots)$,
and $\bs P_n$ is an $\infty\times\infty$ (block) matrix composed by
$\bs S^{(i,j)}$, $\bs I_{\mathcal{N}}$ and $\bs 0$ (see equation \eqref{Pn}).
The compatibility of equation \eqref{Phi-n} and equation $\bs\Phi_{x_{n+1}}=\bs\Phi\bs A \bs P_{n+1}$,
i.e. $(\bs\Phi_{x_n})_{x_{n+1}}=(\bs\Phi_{x_{n+1}})_{x_{n}}$,
gives rise to
\begin{align}\label{com_eq}
    \bs P_{n+1}\bs A\bs P_n-\bs P_n\bs A\bs P_{n+1}+(\bs P_n)_{x_{n+1}}-(\bs P_{n+1})_{x_n}=\bs 0.
\end{align}
Then we have the following.

\begin{proposition}\label{P-2-3}
The $(-1,-n)$-th element of the (block matrix) equation \eqref{com_eq} gives rise to equation \eqref{diff_recur}.
\end{proposition}

\begin{proof}
Consider the case $n\in\mb Z^+$.  $\bs P_n$ is the following block matrix
\begin{equation}\label{Pn}
    \bs P_n=
    \left(\begin{array}{cc:c:cc}
         & \vdots & \vdots & \vdots &  \\
        \cdots & \bs0 & \bs0 & \bs0 & \cdots \\
        \hdashline
        \cdots & -\bs S^{(-1,n-1)} & -\bs S^{(0,n-1)} & -\bs S^{(1,n-1)} & \cdots \\
        \hdashline
        \cdots & -\bs S^{(-1,n-2)} & -\bs S^{(0,n-2)} & -\bs S^{(1,n-2)} & \cdots \\
        \vdots & \vdots & \vdots & \vdots & \vdots \\
        \cdots & \boxed{\bs I_{\mathcal{N}}-\bs S^{(-1,0)}} & -\bs S^{(0,0)} & -\bs S^{(1,0)} & \cdots \\
        \cdots & \bs 0 & \bs I_{\mathcal{N}} & \bs 0 & \cdots \\
        \cdots & \bs 0 & \bs 0 & \bs I_{\mathcal{N}} & \cdots \\
                  & \vdots & \vdots & \vdots &
    \end{array}\right),
\end{equation}
    where the center, i.e. the  $(0,0)$-th element, is $-\bs S^{(0,n-1)}$,
    and the $(-1,-n)$-th element (with a frame) is $\bs I_{\mathcal{N}}-\bs S^{(-1,0)}$.
    $\bs P_{n+1}$ reads
    \begin{align*}
    \bs P_{n+1}=
    \left(\begin{array}{cc:c:cc}
         & \vdots & \vdots & \vdots &  \\
        \cdots & \bs0 & \bs0 & \bs0 & \cdots \\
        \hdashline
        \cdots & -\bs S^{(-1,n)} & -\bs S^{(0,n)} & -\bs S^{(1,n)} & \cdots \\
        \hdashline
        \cdots & -\bs S^{(-1,n-1)} & -\bs S^{(0,n-1)} & -\bs S^{(1,n-1)} & \cdots \\
        \vdots & \vdots & \vdots & \vdots & \vdots \\
        \cdots & \boxed{-\bs S^{(-1,1)}} & -\bs S^{(0,1)} & -\bs S^{(1,1)} & \cdots \\
        \cdots & \bs I_{\mathcal{N}}-\bs S^{(-1,0)} & -\bs S^{(0,0)} & -\bs S^{(1,0)} & \cdots \\
        \cdots & \bs 0 & \bs I_{\mathcal{N}} & \bs 0 & \cdots\\
         & \vdots & \vdots & \vdots &
    \end{array}\right),
    \end{align*}
    where the  $(-1,-n)$-th element (with a frame) is $-\bs S^{(-1,1)}$.
Note that the relation \eqref{diff_sum_pro} with $(i,j)=(-1,0)$ gives rise to $\bs S^{(-1,1)}=\bs U\bs V$,
where $\bs U$ and $\bs V$ are defined in \eqref{UV}.
Thus, the $(-1,-n)$-th element of the relation \eqref{com_eq} can be denoted as
    \begin{align}\label{P-elem}
        \big[\bs P_{n+1}\bs A\bs P_n\big]_{(-1,-n)}-\big[\bs P_{n}\bs A\bs P_{n+1}\big]_{(-1,-n)}
        +\bs U\bs V_{x_n}=
        -(\bs V_{x_{n+1}}+\bs U_{x_n}\bs V).
    \end{align}
Next, we are going to prove the left-hand side of \eqref{P-elem} vanishes.
In fact, by calculation we find
    \begin{align*}
        \big[\bs P_{n+1}\bs A\bs P_n\big]_{(-1,-n)}
        =&\bs S^{(0,1)}\bs a\bs S^{(-1,n-1)}+\bs S^{(1,1)}\bs a\bs S^{(-1,n-2)}\\
        &\qquad\qquad  +\cdots+\bs S^{(n-1,1)}\bs a\bs S^{(-1,0)}-\bs S^{(n-1,1)}\bs a, \\
        \big[\bs P_{n}\bs A\bs P_{n+1}\big]_{(-1,-n)}
        =&\bs S^{(0,0)}\bs a\bs S^{(-1,n)}+\bs S^{(1,0)}\bs a\bs S^{(-1,n-1)}\\
        &\qquad\qquad +\cdots+\bs S^{(n,0)}\bs a\bs S^{(-1,0)}-\bs S^{(n,0)}\bs a,
    \end{align*}
which yields
    \begin{align*}
        &~\big[\bs P_{n+1}\bs A\bs P_n\big]_{(-1,-n)}-\big[\bs P_{n}\bs A\bs P_{n+1}\big]_{(-1,-n)} \\
        =&-\bs S^{(0,0)}\bs a\bs S^{(-1,n)}
        +\sum^{n-1}_{l=0}(\bs S^{(l,1)}-\bs S^{(l+1,0)})\bs a\bs S^{(-1,n-1-l)}
        +(\bs S^{(n,0)}-\bs S^{(n-1,1)})\bs a.
    \end{align*}
Making use of relation \eqref{diff_sum_pro} with $(i,j)=(l,0)$ and $(i,j)=(n-1,0)$, respectively,
and picking out $\bs S^{(0,0)}=\bs U$,  we have
    \begin{align*}
        &~\big[\bs P_{n+1}\bs A\bs P_n\big]_{(-1,-n)}-\big[\bs P_{n}\bs A\bs P_{n+1}\big]_{(-1,-n)} \\
        =&\bs U(-\bs a\bs S^{(-1,n)}+\bs S^{(n-1,0)}\bs a-\sum_{l=0}^{n-1}\bs S^{(n-1-l,0)}\bs a\bs S^{(-1,l)})
        =-\bs U\bs V_{x_n},
    \end{align*}
where we have used \eqref{Sij_moves1} for $(i,j)=(-1,0)$.
Thus, the left-hand side of \eqref{P-elem} vanishes and we get equation \eqref{diff_recur} for $n\in \mathbb{Z}^+$.

The equation \eqref{diff_recur} for $n=0$ and $n\in \mathbb{Z}^-$ can be derived in a similar way.
Once \eqref{diff_recur} is obtained, we immediately get \eqref{SDYM-3}.

\end{proof}

\section{Solutions to a 3-dimensional relativistic-invariant system}\label{APP-B}

In \cite{MZ-1981} Manakov and Zakharov studied the following 3-dimensional
relativistic-invariant system,
\begin{equation}\label{chiral_eq}
    (J_{\bar y}J^{-1})_y+(J_tJ^{-1})_t=0,
\end{equation}
where $t\in \mathbb{R}$, $y\in \mathbb{C}$, $\bar y$ is the complex conjugate of $y$,
and $J$ is a $\mathcal{N} \times \mathcal{N}$ matrix function of $(y, \bar y, t)$.
The system is also known as a $(2+1)$-dimensional generalization of the chiral field equation
(see \cite{Ward1-1988,Ward2-1988}).
To have a solution of the equation \eqref{chiral_eq},
one can take $n$ to be large enough and all $k_j$
subject to $k_j^{2(n+1)}=(-1)^{n}$ such that in \eqref{cj} there are
\[\bs c_j= \left(k_j^n-\frac{(-1)^{n}}{k_j^{n}}, ~k_j^{n+1}+\frac{(-1)^{n}}{k_j^{n+1}}, ~
    \mr i k_j^n+\frac{\mr i(-1)^{n} }{k_j^{n}}, ~0\right)\]
and all $\{\bs c_j\}$ are distinct for $j=1,2,\cdots, N$.
Then the resulting solution $\bs V$ (in both asymmetric and symmetric cases) does not depend on $\eta_{n+1}$
and $J=\bs V$ solves \eqref{chiral_eq}.

\end{appendices}

\end{document}